\documentclass[journal]{IEEEtran}

\usepackage{bm}
\usepackage[T1]{fontenc}
\usepackage{amsthm}
\usepackage{amssymb}   
\usepackage{booktabs}              
\usepackage{algorithm}  
\usepackage{threeparttable}

\usepackage{footnote} 
\usepackage{comment} 
\usepackage{tablefootnote}  

\usepackage{hyperref}
\hypersetup{colorlinks,allcolors=black}

\newtheorem{lemma}{Lemma}

\newtheorem{prop}{Proposition}

\usepackage{cite}

  \usepackage{graphicx}
\usepackage{amsmath}
\usepackage{algorithmic}

\usepackage{array}

\ifCLASSOPTIONcompsoc
\usepackage[caption=false,font=normalsize,labelfont=sf,textfont=sf]{subfig}
\else
\usepackage[caption=false,font=footnotesize]{subfig}
\fi
\usepackage{url}

\begin{document}
\bstctlcite{IEEEexample:BSTcontrol}
\title{Revisiting the Concrete Security of Goldreich's Pseudorandom Generator}
%
%
%

\author{Jing~Yang,~Qian~Guo,~Thomas~Johansson,~and~Michael~Lentmaier
\thanks{Jing Yang, Qian Guo, Thomas Johansson and Michael Lentmaier are with the Department of Electrical and Information Technology, Lund University, Sweden.}}

\maketitle

\begin{abstract}
 Local pseudorandom generators are a class of fundamental cryptographic primitives having very broad applications in theoretical cryptography. 
Following Couteau et al.'s work in ASIACRYPT 2018, this paper further studies the concrete security of one important class of local pseudorandom generators, i.e., Goldreich's pseudorandom generators.
Our first attack is of the guess-and-determine type. Our result significantly improves the state-of-the-art algorithm proposed by Couteau et al., in terms of both asymptotic and concrete complexity, and breaks all the challenge parameters they proposed.
For instance, for a parameter set suggested for 128 bits of security, we could solve the instance faster by a factor of about \(2^{61}\), thereby destroying the claimed security completely.
Our second attack further exploits the extremely sparse structure of the predicate $P_5$ and combines ideas from iterative decoding. This novel attack, named guess-and-decode, substantially improves the guess-and-determine approaches for cryptographic-relevant parameters. All the challenge parameter sets proposed in Couteau et al.'s work in ASIACRYPT 2018 aiming for 80-bit (128-bit) security levels can be solved in about \(2^{58}\) (\(2^{78}\)) operations. We suggest new parameters for achieving 80-bit (128-bit) security with respect to our attacks. We also extend the attack to other promising predicates and investigate their resistance.


\end{abstract}

\begin{IEEEkeywords}
Goldreich's pseudorandom generators, guess-and-determine, guess-and-decode, iterative decoding, \(P_{5}\). 
\end{IEEEkeywords}

%
\IEEEpeerreviewmaketitle

\section{Introduction}
Pseudorandom generators (PRGs) are one fundamental construction in cryptography, which derive a long pseudorandom output string from a short random string. One particular interesting question about PRGs is the existence in complexity class \texttt{NC}$^0$, i.e., each output bit depends on a constant number of input bits. Such constructions, named local pseudorandom generators, can be computed in parallel with constant-depth circuits, thus being highly efficient. A considerable research effort has been devoted to this problem. 

In 2000, Goldreich suggested a simple candidate one-way function based on expander graphs~\cite{goldreich2000candidate}, which has inspired a promising construction for PRGs in \texttt{NC}$^0$. It is constructed as below: given a secret seed $x$ of length $n$ and a well chosen predicate $P$ with locality $d(n)$: $\{0, 1\}^{d(n)} \mapsto \{0, 1\}$, choose $m$ subsets  $(\sigma^1, \sigma^2,...,\sigma^m)$, where each subset contains $d(n)$ disjoint indices of $x$ which are chosen randomly and independently. Let $x[\sigma]$ denote the subset of the bits of $x$ indexed by $\sigma$ and by applying $P$ on $x[\sigma]$, one output bit $P(x[\sigma])$ is obtained. The output string is generated by applying $P$ to all the subsets of bits of $x$ indexed by the sets $(\sigma^1, \sigma^2,...,\sigma^m)$, i.e., the output string is $P(x[\sigma^1])||P(x[\sigma^2])|| ...||P(x[\sigma^m])$. Goldreich advocates $m = n$ and depth $d(n)$ in $O(\log n)$ or $O(1)$, and conjectures that it should be infeasible to invert such a construction for a well-chosen predicate $P$ in polynomial time. The case of $d(n)$ in $O(1)$, which puts the construction into the complexity class \texttt{NC}$^0$, has received more attention due to the high efficiency.

Cryan and Miltersen first considered  the existence of PRGs in \texttt{NC}$^0$ \cite{cryan2001pseudorandom} and gave some results: they applied statistical linear tests on the output bits and ruled out the existence of PRGs in \texttt{NC}$_3^0$ (i.e., each output bit depends on three input bits) for $m \geq 4n$. 
Mossel et al. further extended this non-existence to \texttt{NC}$_4^0$ for $m \geq 24n$ with a polynomial-time distinguisher, but provided some positive results for PRGs in \texttt{NC}$_5^0$ \cite{mossel2003epsilon}.  Specifically, they gave a candidate PRG in \texttt{NC}$_5^0$ instantiated on a degree-2 predicate, which is usually called $P_5$ defined as: 
$$P_5(x_1,x_2, x_3, x_4,x_5) = x_1 \oplus x_2 \oplus x_3 \oplus x_4 x_5,$$
with superlinear stretch while exponentially small bias. Such local PRGs are now commonly known as Goldreich's PRGs, and the ones instantiated on $P_5$ achieve the best possible locality and have received much attention. The existence of PRGs in \texttt{NC}$^0$ (as low as \texttt{NC}$_4^0$) was essentially confirmed in \cite{applebaum2006cryptography} by showing the possibility of constructing low-stretch ($m = O(n)$) PRGs through compiling a moderately easy PRG using randomized encodings. Applebaum et al. in \cite{applebaum2008pseudorandom} further gave the existence of PRGs with linear stretches by showing that the existence can be related to some hardness problems in, e.g., Max 3SAT (satisfiability problem).   

Other than the initial motivation for the efficiency reasoning, i.e., realizing cryptographic primitives that can be evaluated in constant time by using polynomially many computing cores, PRGs in \texttt{NC}$^0$ with polynomial stretches i.e., $m = \text{poly}(n)$, have numerous more emerging theoretical applications, such as {\em secure computation with constant computational overhead}~\cite{ishai2008cryptography,applebaum2017secure}, {\em indistinguishability obfuscation (iO)}~\cite{lin2017indistinguishability,gay2020indistinguishability}, {\em Multiparty-Computation (MPC)-friendly primitives}~\cite{meaux2016towards,grassi2016mpc,albrecht2015ciphers,canteaut2018stream}, {\em cryptographic capsules}~\cite{boyle2017homomorphic}, etc.  For example, a two-party computation protocol with constant computational overhead was given in ~\cite{ishai2008cryptography}, on the assumption of the existence of a PRG in \texttt{NC}$^0$ with a polynomial stretch,  together with an arbitrary oblivious transfer protocol. Thus, the  existence of poly-stretch PRGs in \texttt{NC}$^0$ is attractive and significant.

In \cite{applebaum2013pseudorandom}, Applebaum considered PRGs with long stretches and low localities and provided the existence of PRGs with linear stretches and weak PRGs with polynomial stretches, e.g., $m = n^s$ for some $s>1$, with a distinguishing gap at most $1 / n^s$. This work was later strengthened in \cite{applebaum2016dichotomy} by showing a dichotomy of different predicates: all non-degenerate predicates yield small-bias generators with output length $m = n^s$ for $s < 1.25$  while degenerate predicates are not secure against linear distinguishers for most graphs. The stretch was later extended to $s < 1.5$  for the special case $P_5$ in \cite{odonnell2014goldreich}.
 

The mentioned works above all focus on checking the existence of potential PRGs in \texttt{NC}$^0$, establishing asymptotic security guarantees for them and exploring  appealing theoretical applications based on them; one main obstacle, however, for these advanced cryptographic primitives towards being practical comes from the lack of a stable understanding on the concrete security of these PRGs. In~\cite{couteau2018concrete}, Gouteau et al. first studied the concrete security of Goldreich's PRGs, especially the important instantiation on the $P_5$ predicate. Specifically, they developed a guess-and-determine-style attack and gave more fine-grained security guarantees for them. In the last part of their presentation at ASIACRYPT 2018, an open problem was raised:

{\em ``Can we improve the security bounds for $P_5$?''} 

In this paper we focus on this open problem and give an affirmative answer. 

Before stating our main cryptanalytic methods and results, we first review the common cryptanalysis techniques against local PRGs.

\subsection{Related Work in Cryptanalysis}
The main cryptanalysis tools for local PRGs include myopic backtracking algorithms, linear cryptanalysis, and algebraic cryptanalysis, etc.
\subsubsection{Myopic Backtracking Algorithms}
A Goldreich's PRG can be viewed as a random constraint satisfaction problem (CSP), thus the inversion of the PRG is equivalent to finding a planted solution for a CSP. Thus, some techniques and results from solving CSPs, e.g., 3-SAT, can be adopted. The so-called myopic backtracking method is such one commonly used algorithm, the basic idea of which is to gradually assign values to some input variables in every step based on newly read $t$ constraints and previous observations until a contradiction is introduced. Every time when the new partial assignments contradict with some constraints, the algorithm backtracks to the latest assigned variable, flips the assigned value and continues the process. After sufficiently many steps, the algorithm will surely recover the secret.

Goldreich considered the myopic backtracking method on the proposed one-way function \cite{goldreich2000candidate}, by reading one output bit at each step and computing all possible values of input bits which would produce the read output bits. The results show that the expected size of the possible values is exponentially large. Alekhnovich et al.  further gave exponential lower bounds of the running time for myopic algorithms in \cite{alekhnovich2006exponential}.
It showed that Goldreich's function is secure against the myopic backtracking algorithm when it is instantiated on a 3-ary predicate $P(x_1, x_2, x_3) = x_1 \oplus x_2 \oplus x_3$. Since a predicate should be non-linear (otherwise a system can be easily solved using Gaussian elimination), the predicate is extended to a more general case involving a non-linear term: $P_d(x_1,...,x_d) = x_1 \oplus x_2 \oplus ... \oplus x_{d-2} \oplus x_{d-1} x_d$. They showed that for most $d$-ary predicates and some $t$, the expected success probability of the basic $t$-myopic algorithm (i.e., reading $t$ constraints at each step) in inverting is $e^{-\Omega(n)}$. The results were verified over small Goldreich's PRGs using the SAT solver MiniSat. 

\subsubsection{Linear Cryptanalysis}
Each constraint of a local PRG can be viewed as a linear equation with a noise, i.e., $y_i = \sum_{\substack{j \in \hat{\sigma}^{i}}} x_j + e_i,$ where $\hat{\sigma}^{i}$ is the subset of linear terms of the $i$-th subset $\sigma^{i}$, while $e_i$ is a biased noise, whose distribution is determined by the chosen predicate. For the $P_5$ predicate, $e_i$ has the distribution $P(e_i = 1) = 1/4, P(e_i = 0) = 3/4$.

One simple way of linear cryptanalysis could be to find equations sharing the same linear variables and view the noises as independent and identically distributed. A majority rule can be applied on the noises and a corresponding value can be assigned to the linear part. Thus a linear equation is obtained and by exploring many such linear equations, the secret could be recovered by solving the derived linear system.

An improved version is to build enough noisy equations with two linear terms of the form $x_{\sigma_i} + x_{\sigma_j} (+ e_i + e_j) = y_i + y_j$, by XORing a pair of equations sharing the other linear variables, and then apply some known algorithms, e.g., semidefinite programming \cite{applebaum2016cryptographic}, to get a solution $\hat{x}$ satisfying a large fraction of the equations. This solution would be highly correlated with the true one and the system could be inverted with high probability based on it \cite{bogdanov2009security}.
\subsubsection{Algebraic Attacks} 
In an algebraic attack, the equation system is extended by, e.g., multiplying some equations with lower-degree terms until a solution is possible to be found by, e.g., linearization, Gaussian elimination or by computing a Gr\"{o}bner basis of the expanded system. In \cite{applebaum2018algebraic}, Applebaum and Lovett analyzed how the underlying predicate affects pseudorandomness using algebraic attacks and gave some advice on the choices of predicates in terms of resiliency, algebraic degree and bit fixing degree.   Algebraic attacks based on linearization and Gr\"{o}bner base algorithms were further considered in \cite{couteau2018concrete}, and some results on concrete choices of parameters were given.


\subsection{Contributions} 
\label{subsec:contributions}  

In this paper, we present new attacks which significantly improve the complexity of inverting the local PRGs instantiated on the $P_5$ predicate.

\begin{itemize}
\item[--] Our first result is a novel guess-and-determine-style attack with much lower complexity than the results presented in~\cite{couteau2018concrete}. 
  We develop theoretical and also numerical analysis about how many guesses that are needed for various $(n,s)$ parameters, where $n$ and $s$  denote the seed size and stretch, respectively, and experimentally verify the analysis for some small parameters. \\

This approach is basically a greedy method. We classify the equations occurring during the guess-and-determine attack into three different classes, by the number of included monomials. The class of equations with two monomials are desired, as guessing variables occurring in these equations could produce some ``free'' determined variables. 
 We then design our guessing criterion over the equation classes, to generate as many ``free'' variables as possible when assuming for a limited number of guesses. \\

Our guess-and-determine attack reduces the asymptotic complexity from $\tilde{O}( 2^{\frac{n^{2-s}}{2}})$ in~\cite{couteau2018concrete} to $\tilde{O}( 2^{\frac{n^{2-s}}{4}})$, thereby achieving a square-root speed-up.
Regarding the concrete complexity, we could solve all the challenge parameters suggested in~\cite{couteau2018concrete} with complexity much below their claimed security levels. As shown in Table~\ref{tab:gain}, our complexity gains (measured by the ratio of the two complexity numbers) range from \(2^{23}\) to \(2^{30}\) for parameters aiming for \(80\) bits of security, and from \(2^{40}\) to \(2^{61}\) for those aiming for \(128\) bits. \\
\item[--] Based on the guessing strategies proposed in the first attack, we further exploit the extremely sparse structure of $P_5$ and combine ideas from iterative decoding for solving a linear system. Our method differs from the classic iterative decoding or belief propagation, since the system in our case is quadratic and no a-priori information is available. 
We design new belief propagation rules for this specific setting.  
  This is a novel method for solving a system of non-linear equations and we call it the {\em guess-and-decode} approach. The gain of it compared to the first attack comes from 1) the better information extraction from the equations with a quadratic term, since the guess-and-determine approach only exploits the generated linear equations; and 2) the soft probabilities used in the iterative decoding, since inverting a linear matrix in the  guess-and-determine approach can be regarded as using ``hard'' (binary) values. \\
  \\  
  Experimental results show that, with a smaller number of guesses, the resulting system has a good chance to be determined and then the secret could be fully recovered. 
  As shown in Table~\ref{tab:gain}, the new guess-and-decode approach could further significantly improve the guess-and-determine approach for all the challenge parameters proposed in~\cite{couteau2018concrete}. For instance, the improvement factor is as large as \(2^{22}\) for the parameter set \((4096,1.295)\), and could be even larger for a parameter with a bigger \(n\) value in our prediction. With this new method, the challenge parameter sets proposed in~\cite{couteau2018concrete} for 128-bit security are insufficient even for providing security of 80 bits. \\
  \\
   We also suggest new challenge parameters to achieve 80-bit and 128-bit security levels for various seed sizes. \\
  \item[--] Lastly, we extend the attacks to other promising predicates of the type of XOR-AND and XOR-MAJ, which are the two main types of predicates suggested for constructing local PRGs. We investigate their resistance against our attacks and give some initial sights on their possibly safe stretches.
\end{itemize}

\begin{table}[!t]
	\renewcommand{\arraystretch}{1.2}
	\centering
	\caption{The Complexity Comparison of the Algorithms for Solving the Challenge Parameters Proposed in~\cite{couteau2018concrete}\tnote{}.}
	\label{tab:gain}
	\begin{threeparttable}
	\begin{tabular}{c|c|r|r|r}
		\toprule
		Security Level & (\(n, s\)) &   \cite{couteau2018concrete} & Sec.~\ref{sec:method} & Sec.~\ref{sec:guess_decode}  \\ \hline
		80 bits & (\(512, 1.120\)) & \(2^{91}\) & \(2^{61}\) & \(2^{52}\) \\
		& (\(1024, 1.215\)) & \(2^{90}\) & \(2^{66}\) & \(2^{53}\) \\
		& (\(2048, 1.296\)) & \(2^{91}\) & \(2^{68}\)  & \(2^{57}\) \\
		& (\(4096, 1.361\)) & \(2^{91}\)  & \(2^{68}\)  & \(2^{58}\) \\
		\hline
		128 bits & (\(512, 1.048\))  & \(2^{140}\) & \(2^{79}\)& \(2^{68}\) \\
		& (\(1024, 1.135\)) & \(2^{140}\) & \(2^{93}\) & \(2^{72}\)  \\
		& (\(2048, 1.222\))  & \(2^{139}\) & \(2^{98}\)& \(2^{77}\) \\
		& (\(4096, 1.295\)) & \(2^{140}\)  & \(2^{100}\)  & \(2^{78}\) \\			 
		\bottomrule	
	\end{tabular}
	\begin{tablenotes}
		\item[]  The column ``Sec.~\ref{sec:method}' shows the complexity of the guess-and-determine attack and the column ``Sec.~\ref{sec:guess_decode}'' shows the complexity of the guess-and-decode attack.  
	\end{tablenotes}
	\end{threeparttable}
\end{table}



\noindent {\bf Discussions} It is pointed out in~\cite{boyle2019efficient} that  ``...it (building cryptographic capsule upon the group-based homomorphic secret sharing together with Goldreich's low-degree PRG in \cite{boyle2017homomorphic}) is entirely impractical: Goldreich's PRG requires very large seeds...'', and the best-known instantiations of some recent attractive  applications such as Pseudorandom Correlation Generators~\cite{boyle2019efficient,DBLP:conf/crypto/BoyleCGIKS20}  are based on computational assumptions such as LPN (learning parity with noise) and Ring-LPN, rather than on Goldreich's low weight PRGs. Our novel attacks further diminish their practicality, though their significance in theoretical cryptography is unaffected. We strongly recommend to instantiate a more complex local PRG with higher locality and/or more non-linear terms, if relatively high stretches and high security are required. Then, the research on the concrete security of the new PRGs should be renewed. However, in some special applications where drastic limits on the number of output bits are not an issue, while the significant advantages from the extremely low locality, multiplicative and overall complexity for deriving output bits are more desired, local PRGs instantiated on $P_5$ would probably still be considered. Our attacks provide better understanding of the concrete security of Goldreich's PRGs in a more fine-grained manner, and shed light on producing better cryptanalytical works on the MPC-friendly primitives with similar inherent structures.
  
\subsection{Organization}
\label{subsec:organization}
The rest of this paper is organized as follows. We first give some preliminaries and briefly describe the guess-and-determine attack in \cite{couteau2018concrete} in Section \ref{sec:preliminaries}. We then present our improved guess-and-determine attack in Section \ref{sec:method}, describing how it works and providing theoretical analysis and experimental verification. In Section \ref{sec:guess_decode}, we thoroughly describe the guess-and-decode attack and verify it by extensive experimental results. We extend the attacks to other promising predicates in Section \ref{sec:extension} and lastly conclude the paper in Section \ref{sec:conclusion}.

    
\section{Preliminaries} \label{sec:preliminaries}
\subsection{Notations}
Let $GF (2)$ denote the finite field with two elements. Thus, the operation of addition ``$+$'' is equivalent to the ``$\oplus$'' operation. Our paper focuses on Goldreich's PRGs instantiated on the $P_5$ predicate. Throughout the paper, we use $x \in \{0, 1\}^n$ to denote the secret seed of size $n$ and $x_i~(1 \leq i \leq n)$ to denote the $i$-th bit of $x$. We use $y \in \{0, 1\}^m$ to denote the output string of size $m$ and $m = n^s$, where $s$ is the expansion stretch. Each bit of $y$, denoted as $y_i$ for $1 \leq i \leq m$, is derived by applying $P_5$ on a 5-tuple subset of the seed variables. These seed variables are indexed by a publicly known index subset $\sigma^i = [\sigma_1^i, \sigma_2^i, \sigma_3^i, \sigma_4^i, \sigma_5^i]$, where indices are distinct and randomly chosen. Thus, $y_i$ is generated as below: 
$$
x_{\sigma_1^i} + x_{\sigma_2^i} + x_{\sigma_3^i} + x_{\sigma_4^i} x_{\sigma_5^i} = y_i.
$$
We call such a relation an equation. 

For a binary variable $u$, we use $p_{u}^y$ or $P(u = y)$ to denote the probabilities of $u$ being the value of $y$, where $y \in \{0, 1\}$. The Log-Likelihood Ratio (LLR) value of $u$ is defined as 
\begin{align}
	L_u = \log \frac{P(u = 0)}{P(u = 1)},
\end{align}
i.e., the logarithmic value of the ratio of probabilities of $u$ being 0 and being 1. All the secret bits are assumed to be uniformly random distributed, thus the initial LLR values for them are all zero.

\subsection{The Guess-and-Determine Attack in \cite{couteau2018concrete}}
In a guess-and-determine attack, some well-chosen variables of a secret are guessed and more other variables could be correspondingly determined according to predefined relationships connected to these guessed variables. If the partial guessed variables are guessed correctly, the further determined variables would be correct as well. Thus, one hopes that by guessing a smaller number than the expected security level of variables, all the other variables can be determined. 

A guess-and-determine attack on Goldreich's PRG is given in \cite{couteau2018concrete}. The basic idea is to guess enough secret variables and derive a system of linear equations involving the remaining variables and by solving this linear system, the secret is expected to be recovered.  

When guessing a variable, those equations involving this guessed variable in the quadratic terms would become linear. The guessing strategy in \cite{couteau2018concrete} is then to always guess the variable which appears most often in the quadratic terms of the remaining quadratic equations, thus to obtain a locally optimal number of linear equations for each guess. This process is iteratively performed until enough linear equations are derived.  

Besides, some ``free'' linear equations can be obtained before the guessing phase by XORing two equations sharing a same quadratic term, which is called a ``collision''. The average number of linear equations $c$ derived in this way has been computed in \cite{couteau2018concrete}. 

The guessing process can be stopped once $n-c$ linear equations are obtained (with the number of guesses included). Suppose that after guessing $\ell$ variables, the stopping condition is satisfied. The algorithm will enumerate over all $2^\ell$ possible assignments for these guessed variables, and for each assignment $e$, a distinct system of linear equations would be derived. Let $A_e$ denote the matrix of this linear system whose rank could be equal to or smaller than $n$. The paper shows that when the rank of $A_e$ is smaller than $n$, an invertible submatrix with fewer variables involved can almost always be extracted and thus a fraction of variables can always be recovered. The remaining variables can be easily obtained by injecting those already recovered ones. 

In \cite{couteau2018concrete}, an upper bound of the expected number of guesses, denoted as $\ell$, is given as $\ell \leq \lfloor \frac{n^2}{2m} + 1  \rfloor$, and the asymptotic complexity is $O(n^2 2^{\frac{n^{2-s}}{2}})$, where $O(n^2)$ is the asymptotic complexity for inverting a sparse matrix. The attack is experimentally verified and some challenge parameters under which the systems are resistant against the attack are suggested.

\section{New Guess-and-Determine Cryptanalysis of Goldreich's PRGs with $P_5$} \label{sec:method}
In this section, we give a new guess-and-determine attack on Goldreich's pseudorandom generators. The fundamental difference with the guess-and-determine attack in \cite{couteau2018concrete} is that the guessing process in our attack is ``dynamic'', by which we mean that the choice of a new variable to guess depends not only on previously guessed variables but also on their guessed values. The main observation for our attack is that some variables could be determined for free after having guessed some variables, and the goal of our attack is to exploit as many such ``free'' variables as possible. To achieve so, we first define three different equation classes, {\bf { Class I, II}} and {\bf III}, which include different forms of equations, and further design guessing strategies based on them. 

\subsection{Equation Classes and ``Free'' Variables} \label{subsec:equation_classes}
We categorize equations generated during the guessing process into three classes according to the number of terms.\\ 
\medskip

\fbox{
\parbox{0.45\textwidth}{
{\bf Class I} includes equations having no less than four terms, either quadratic or linear, with forms as below:
\begin{align}
x_{\sigma_1^i} + x_{\sigma_2^i} + x_{\sigma_3^i} + x_{\sigma_4^i}x_{\sigma_5^i} = y_i  \label{eq:Class1_1}, \\
x_{\sigma_1^j} + x_{\sigma_2^j} + x_{\sigma_3^j} + x_{\sigma_4^j} = y_j \label{eq:Class1_2}.
\end{align}
}}
\medskip

\noindent
Equations with more than four terms, which can be generated, e.g., from XORing two equations sharing a  same  quadratic term, are also categorized into Class I. The equations with the form in \eqref{eq:Class1_1} are the initially generated equations. If one variable in the quadratic term of such an equation, e.g., $x_{\sigma_4^i}$ or $x_{\sigma_5^i}$, is guessed as 1, the equation would be transformed to an equation in \eqref{eq:Class1_2}. While if the guessed value is 0 or a variable in the linear terms is guessed, we would get an equation categorized into another class, Class II. \\

\medskip

\fbox{
\parbox{0.45\textwidth}{
\noindent {\bf Class II} includes those equations having three terms, either quadratic or linear, with forms as below:
\begin{align}
x_{\sigma_1^i} + x_{\sigma_2^i} + x_{\sigma_3^i}x_{\sigma_4^i} = y_i \label{eq:Class2_1},\\ 
x_{\sigma_1^j} + x_{\sigma_2^j} + x_{\sigma_3^j} = y_j \label{eq:Class2_2} . 
\end{align}
}}

\medskip

\noindent
As mentioned before, an equation with the form in \eqref{eq:Class2_2} could be obtained if one variable in the quadratic term of an equation in \eqref{eq:Class1_1} is guessed as 0. It can also be derived when one arbitrary variable in an equation in \eqref{eq:Class1_2} is guessed (either 1 or 0), or if one variable in the quadratic term of an equation in \eqref{eq:Class2_1} is guessed as 1. Equations with the form in \eqref{eq:Class2_1} can be derived by guessing one linear variable of an equation in \eqref{eq:Class1_1}. \\

\medskip

\fbox{
\parbox{0.45\textwidth}{
\noindent {\bf Class III} includes those equations having only two terms, either quadratic or linear,  with forms as below:
\begin{align}
x_{\sigma_1^i} + x_{\sigma_2^i}x_{\sigma_3^i} = y_i \label{eq:Class3_1},\\ 
x_{\sigma_1^j} + x_{\sigma_2^j}  = y_j \label{eq:Class3_2}.
\end{align}
}}

\medskip

\noindent
If one linear term of an equation in \eqref{eq:Class2_1} is guessed, it will be transformed to an equation with the form of \eqref{eq:Class3_1}. An equation with the form of \eqref{eq:Class3_2} can be derived from: guessing one variable in the quadratic term in \eqref{eq:Class2_1} as 0; or guessing one variable in \eqref{eq:Class2_2} (either 1 or 0); or guessing one variable in the quadratic term in \eqref{eq:Class3_1} as 1. \\

\medskip

If we guess a variable which is involved in an equation in Class III, some ``free'' information could be obtained, which can happen in the following cases:
\begin{itemize}
	\item[(1)] For a linear equation $x_{\sigma_1^i} + x_{\sigma_2^i}  = y_i$, when one variable, e.g., $x_{\sigma_1^i}$, is guessed, the other variable, $x_{\sigma_2^i}$ in this case, can be directly derived as $x_{\sigma_2^i}  = x_{\sigma_1^i} + y_i$ for free;	
	\smallskip
	\item[(2)] For a quadratic equation $x_{\sigma_1^i} + x_{\sigma_2^i}x_{\sigma_3^i} = y_i$, if the linear variable $x_{\sigma_1^i}$ is guessed to be different from $y_i$, $x_{\sigma_2^i}$ and $x_{\sigma_3^i}$ could be easily derived as $x_{\sigma_2^i} = 1, x_{\sigma_3^i} = 1$, thus two  ``free'' variables are obtained;
	\smallskip
	\item[(3)] If any of the variables in the quadratic term is guessed as 0 in the above quadratic equation, the linear term is derived to be $x_{\sigma_1^i} = y_i$ for free;
	\smallskip
	\item[(4)] If $x_{\sigma_1^i}$ is guessed as the value of $y_i$ in the above quadratic equation, an equation with only one quadratic term would be obtained, i.e., $x_{\sigma_2^i}x_{\sigma_3^i} = 0$. There are no ``free'' variables obtained at the moment. If at a future stage any one of the two variables, $x_{\sigma_2^i}$ or $x_{\sigma_3^i}$, is guessed to be 1, the other one can be derived as 0 for free. 
\end{itemize}

Thus, the criterion for the guessing is always guessing variables in equations from Class III if it is not empty and thus getting  ``free'' variables. Every time after guessing one variable, some equations in Class I and Class II could be transformed to equations in Class II or Class III. So it is highly likely that there always exist equations in Class III to explore for ``free'' variables. Furthermore, after plugging in the ``free'' variables, it could happen that more ``free'' variables would be derived. Thus, the required number of guesses can be largely reduced. That is the motivation of categorizing equations into three different classes and we next describe the new guess-and-determine attack designed over the defined equation classes.
\subsection{Algorithm for the New Guess-and-Determine Attack} \label{subsec: guessing_algorithm}
Algorithm \ref{alg:guess_and_determine} gives the general process of the proposed guess-and-determine attack. For convenience of description, we use the term {\it {equation reduction}} to denote the process of plugging in the value of a variable, either guessed or freely derived, and transforming the equations with the variable involved to ones with lower localities. Below we give the details of the attack.

\begin{algorithm}[!t] 
	\caption{The new guess-and-determine attack} 	
	\textbf{Input} 
	$m$ quadratic equations derived from a length-$n$ secret according to the $P_5$ predicate
	\\	
	\textbf{Output} 
	the secret
	\begin{algorithmic}[1] 	\label{alg:guess_and_determine}
		\STATE Exploring linear equations obtained through collisions \label{step:collisions}
		\STATE Set $t = 0$, back up the current system, mark all variables as ``non-reversed''
		\STATE Guess a variable and perform equation reduction following Algorithm \ref{alg:guess}, $t = t + 1$, back up the derived system\label{step:guess}
		\IF {no less than $n$ linear equations are derived}
			\STATE solve the linear system
			\IF {the recovered secret is correct}
				\RETURN the recovered secret
			\ELSE 
				\STATE trace back: \label{step:trace_back}
				\IF {the $t$-th guessed variable has been marked as ``reversed''}
					\STATE delete the $t$-th back-up system, recover the variable as ``non-reversed''
					\STATE $t = t -1$, go to step 9
				\ELSE 
					\STATE delete the $t$-th back-up system, mark the variable as ``reversed''
					\STATE reverse the guessed value of the variable and go to step 3 over the $(t-1)$-th back-up system without increasing $t$
				\ENDIF
			\ENDIF
		\ELSE 
			\STATE go to step 3
		\ENDIF
	\end{algorithmic}
\end{algorithm}

Before guessing, as done in \cite{couteau2018concrete}, we first explore linear equations obtained through XORing two equations sharing a same quadratic term. The quadratic terms would be canceled out and only the linear terms of the two equations remain. If these two equations further share one or two linear terms, the shared variables should be canceled out as well. Thus, linear equations derived in this way could have two, four, or six variables, which are put into Class III, Class I and Class I, respectively. Note that when any two equations share the same variables in the linear and quadratic part, respectively, it already leads to a distinguishing attack. The expected number of linear equations derived in this way has been given in \cite{couteau2018concrete}. 

After finding out the linear equations derived from collisions in the quadratic terms, we back up the current system and start the guessing process. We follow the rules in Algorithm \ref{alg:guess} to choose variables for guessing. Specifically, we always choose a variable in an equation in Class III, since it is more likely to obtain ``free'' variables. If Class III is empty, we guess a variable in an equation from a less attractive class, Class II, or Class I if  Class II is also empty. As for choosing which variable to guess, we use the following criterion: we always choose the variable appearing the most number of times in the local class, e.g., Class III if a variable in an equation in Class III is guessed, since it could introduce more  ``free'' variables or transform more equations; if all variables have the same occurrences locally, choose the variable which appears most often in the global  system of equations, thus more equations would be transformed.

\begin{algorithm}[!t] 
	\caption{Algorithm for guessing a variable} 
	\textbf{Input} 
	A system of equations
	\begin{algorithmic}[1] \label{alg:guess}
		\IF {{\bf Class III} is not empty}
		\STATE guess the variable appearing most often in {\bf Class III}
		\STATE plug in the guessed variable and perform equation reduction \label{step:class3}\\
		\STATE continually plug in  ``free'' variables and perform equation reduction, if there are any, until there are no more ``free'' variables
		\ELSIF {{\bf Class II} is not empty}
		\STATE guess the variable appearing most often in {\bf Class II} \label{step:class2}
		\STATE plug in the guessed variable and perform equation reduction
		\ELSE  
		\STATE guess the variable appearing most often in {\bf Class I} \label{step:class1}
		\STATE plug in the guessed variable and perform equation reduction			
		\ENDIF
	\end{algorithmic}
\end{algorithm}

Every time when one variable is guessed, an equation involving this guessed variable would be transformed to either a linear equation in the same class or an equation (could be quadratic or linear) in the next class. The goal is to transform as many equations as possible to Class III such that when guessing one variable in Class III, it is more likely to get some ``free'' variables. After each guess and corresponding equation reduction, we should also plug in the ``free'' variables, if there are any, and perform equation reduction. This step might introduce more ``free'' variables and we continue this process until there are no further  ``free'' variables.

Every time after performing equation reduction for a guessed and corresponding freely determined variables, we always back up the derived system, in case the guesses are not correct and at some future stage, tracing back is needed. If there are enough linear equations derived after some guesses, we would stop guessing and solve the linear system. By ``enough'' we use the condition in \cite{couteau2018concrete} for key recovery, i.e., the number of linear equations (including the guessed and freely determined variables and linear equations derived by finding collisions in quadratic terms) is not smaller than $n$. The rank of the matrix for the linear system does not have to be $n$, as it shows in \cite{couteau2018concrete} that an invertible subsystem with fewer variables involved can almost always be extracted and solved. The remaining small fraction of variables can be easily recovered by injecting those already recovered ones.    

If the recovered secret is correct, which can be verified by checking if it can produce the same output sequence, the attack succeeds and stops. While if not, we need to trace back: either 1) reversing the guessed value of the last guessed variable if it has not been reversed before and perform equation reduction based on the last back-up system; or 2) tracing back more steps until to a guessed variable which is not reversed before. Since we have backed up in each stage, we can retrieve the desired system of equations, perform equation reduction over that system, and delete all subsequent back-ups when we trace back. If some conflicts happen during the reduction, for example, we have guessed a variable as 1 while some other equations give the information that this variable should be 0, we stop going into deeper and immediately trace back. 

\begin{figure}[!ht]
	\centering
	\includegraphics[width=0.48\textwidth]{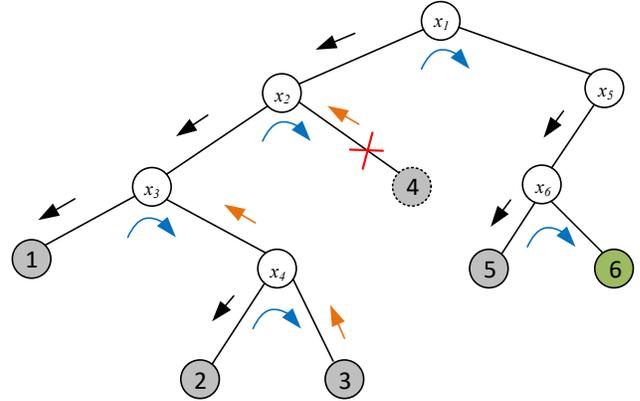}
	\medskip
	\caption{A  simple illustration of the guessing process}
	\label{fig:guess_process}
\end{figure}

The path of the guessing is a binary tree while the tree is highly likely to be irregular, i.e., the depths for different paths could be different, but with small gaps. Fig. \ref{fig:guess_process} shows a simple illustration of the guessing process. The leaf nodes, which are color-filled, represent the process to solve a system of linear equations when enough linear equations are derived. The gray-filled ones indicate the wrong paths, i.e., some guesses are not correct, while the green-filled one is the correct one and the guessing stops at this node. Black arrows indicate normal guessing steps, while blue arrows indicate the processes of reversing the guessed values of the last guessed variables and orange arrows indicate the processes of tracing back. At the leaf node with index 4, some conflicts happened during the equation reduction, which indicates that some previous guessed values are not correct, e.g., $x_1$ is not correctly guessed in this case. So we need to reverse the guessed value of it and continue guessing. In this simple example, we have visited  six nodes, meaning that we have guessed six variables, and solved four systems of linear equations until we find the correct path. We could then get the correct values of the variables $x_1, x_5, x_6$ as (1, 0, 1) (we denote the left direction  as guessing 0 and right direction as guessing 1), and other variables can be recovered by solving the derived linear system at the leaf node with index 6.

In practice, guessing 0 or 1 has a negligible impact on the attack complexity for the wrong paths, as both values would be tried. However, for the correct path, for an equation in Class III, e.g., $x_{\sigma_1^i} + x_{\sigma_2^i}x_{\sigma_3^i} = y_i$, $x_{\sigma_1^i}$  is more likely to equal $y_i$ than the complement. Thus we set the criterion that when guessing the linear term in a quadratic equation in Class III, we always first guess it to be equal to the value of the equation.

\subsection{Theoretical analysis}
In this subsection, we derive theoretical analysis on the required number of guesses and the complexity of the attack. 

\begin{prop}
	{(\bf Number of guesses)}. For any instance with $n$ variables, $n^s$ equations and $c$ collisions, the average number of guesses required to build $n$ linear equations is:
	\begin{equation} \label{eq:theoretical}
		\ell =  \left\lceil \frac{n -c}{4n^{s-1} + 2} + 1 \right\rceil.
	\end{equation}
\end{prop}
\begin{proof}
First we derive the average ``free'' variables one can get by guessing one variable. 

The average occurrence of a variable in the initial system of equations can be calculated as $5n^{s-1}$ (since we have $n^s$ equations and thus $5n^{s}$ occurrences of variables, we can get the average occurrence of one variable is ${5n^{s}}/{n}$), and further divided into $2n^{s-1}$ and $3n^{s-1}$ in the quadratic and linear terms, respectively. 

For the first variable to guess (we assume that Class II and Class III are empty when guessing the first variable and if they are not, we have even more advantages), if it appears as a linear term in an equation, this equation would be reduced to a quadratic equation in Class II (no matter the variable is guessed to be 1 or 0). The average number of such equations would be $3n^{s-1}$. If the variable appears in the quadratic term of an equation, the equation would be reduced to a linear equation either in Class I (when the guessed value is 1) or in Class II (when the guessed value is 0). The average number of such equations would be $2n^{s-1}$. Thus after the first guess, we have on average $3n^{s-1}$ quadratic equations in Class II and $2n^{s-1}$ linear equations either in Class I or Class II. 

Similarly, for the second guess from equations in Class II, roughly another $2n^{s-1}$ linear equations would be derived. Thus, after the second guess, we have around $4n^{s-1}$ linear equations. At the same time, some equations in Class II would be reduced to equations belonging to Class III. From the third guess, it is highly likely that there are always equations in Class III and we can keep guessing variables appearing in this class to explore ``free'' variables. There are three different cases in terms of choosing a variable and obtaining ``free'' variables:
\begin{itemize}
	\item[(1)] {\bf Guessing a variable in a linear equation.} For a linear equation $x_{\sigma_1^i} + x_{\sigma_2^i} = y$, if one variable, say $x_{\sigma_1^i}$ is guessed, the other variable could be freely determined, i.e., $x_{\sigma_2^i} = y + x_{\sigma_1^i}$. By plugging in $x_{\sigma_2^i}$ and performing equation reduction, it might happen that more variables can be further determined. Thus, at least one ``free'' variable would be obtained in this case. If we plug in the guessed and ``free'' variables, we would  get roughly $4n^{s-1}$ linear equations on average for every guess.
	\smallskip
	\item[(2)] {\bf Guessing the linear variable in a quadratic equation.} For a quadratic equation $x_{\sigma_1^i} + x_{\sigma_2^i}x_{\sigma_3^i} = y_i$, if $x_{\sigma_1^i}$ is guessed as $x_{\sigma_1^i} = y_i \oplus 1$, we could get two ``free'' variables, i.e., $ x_{\sigma_2^i} =1, x_{\sigma_3^i} =1$; while if it is guessed as $y_i$, no ``free'' variables could be obtained from the quadratic equation $x_{\sigma_2^i}x_{\sigma_3^i} = 0$ for the moment. If we assume the variable is guessed as 1 or 0 randomly, one ``free'' variable on average can be obtained by guessing the linear term. However, for the correct path, $x_{\sigma_1^i}$ is more likely to be $y_i$: with probability of 0.75, instead of 0.5, and we would get less than one ``free'' variable.
	\smallskip
	\item[(3)] {\bf Guessing a variable in the quadratic term.} When guessing a variable in the quadratic term in the above quadratic equation, if the variable is guessed to be 0, $x_{\sigma_1^i} = y_i$ could be obtained for free; while if it is guessed to be 1, we could get a linear equation in Class III. If one more variable in this derived linear equation is guessed in some future stage, one ``free'' variable could be achieved. Therefore, by guessing a variable in the quadratic term, we could get $2/3$ free variables on average.	
\end{itemize}

Thus, when guessing a variable in Class III, we set the rule that we always first guess variables from linear equations, to guarantee at least one ``free'' variable for every guess. From the third guess, there exist linear equations in Class III with very high probability. When we have to guess a variable from a quadratic equation when no linear equations exist, we always guess the linear term, such that the average number of ``free'' variables is one (except for the correct path). This enables us to reach the stop condition and trace back as fast as possible for the wrong  paths. Thus, we could get $4n^{s-1}$ linear equations for each guess on average, and after $k$ guesses, we would roughly get $k-2$ ``free'' variables and $4(k-1)n^{s-1}$ linear equations.


Thus, after $k$ guesses, the number of linear equations (including the guessed and freely determined variables) one can obtain is roughly
\begin{equation} \label{eq:number}
	4(k-1)n^{s-1} + c + 2k - 2.
\end{equation}
The average value of $c$ has been derived as $m - \binom{n}{2} + \binom{n}{2}((\binom{n}{2} -1)/\binom{n}{2})^m$ in \cite{couteau2018concrete}. We use the same stop condition as that in \cite{couteau2018concrete}, i.e., the number of linear equations  is not smaller than $n$. Suppose that after guessing $\ell$ variables, it will for the first time satisfy $4(\ell-1)n^{s-1} + c +  2l -2  = n$, then we can get $\ell =  \lceil \frac{n -c}{4n^{s-1} + 2} + 1 \rceil$.

\end{proof}

\noindent {\bf Storage Complexity.} We need to back up several intermediate systems generated during the guessing process. The number of nodes with back-up is the depth of a guessing path, i.e., the number of guesses $\ell$. Thus the storage complexity is $O(\ell\cdot m)$, which is $O(n^2)$ according to \eqref{eq:theoretical}.\\
\\
\noindent {\bf Time Complexity.} When selecting a variable to guess, we need to go through all the equations in a local class, e.g., Class III, or Class II if Class III is empty, and choose the variable which appears most often. If these variables have the same occurrences in the local classes, we choose the one among them which appears most often in the global system. We at most need to go through all the equations and the complexity is upper bounded by $O(m)$. We could further reduce the complexity by keeping a global list recording the occurrence of each variable and updating it when performing equation reduction. Thus, the global maximum value can be found in complexity $O(n)$.
The total number of selections is $O(2^\ell)$ and thus the complexity for selecting variables is $O(2^\ell n)$, which is $O(n 2^{\frac{n^{2-s}}{4}})$. 

The largest computation overhead lies in solving linear systems, for which the cost is dominated by inverting a matrix. We use the same estimation of time complexity as that in \cite{couteau2018concrete} for inverting a sparse matrix, which is $O(n^2)$. Thus the total time complexity is dominated by $O(2^\ell \cdot n^2)$, i.e., $O(n^2 2^{\frac{n^{2-s}}{4}})$ according to \eqref{eq:theoretical}, which has halved the exponential coefficient of the time complexity $O(n^2 2^{\frac{n^{2-s}}{2}})$ in \cite{couteau2018concrete}.

\begin{lemma}
	The asymptotic complexity of the proposed guess-and-determine attack is
	\begin{equation*}
	O(n^2 2^{\frac{n^{2-s}}{4}}).
	\end{equation*}
\end{lemma}
 
\subsection{Experimental Verification}
In this section, we first experimentally verify the attack and the theoretical analysis, then break the candidate non-vulnerable parameters suggested in \cite{couteau2018concrete}. 
We first use a simpler way with much less complexity to verify the attack and later show that the results derived in this way match well with those for practically recovering the secret. In this simple verification, we test the required number of guesses to collect enough linear equations, i.e., the number of linear equations is not smaller than $n$. We iteratively select a variable using the criterion described above and guess it randomly until the condition is satisfied. We run 1040 instances for each $(n, s)$ pair and get the average number of guesses, whose distribution shows a low variance. The theoretical results computed according to \eqref{eq:theoretical} and experimental results of the average required number of guesses are shown in Table \ref{tab:combined}. The experimental results are better than the theoretical ones, just as happened in \cite{couteau2018concrete}. Compared to the experimental results in \cite{couteau2018concrete}, the proposed attack requires fewer guesses. 

\begin{table*}[!ht]
		\renewcommand{\arraystretch}{1.2}
	\centering
		\caption{The Average Number of Guesses Required to Achieve Enough Linear Equations\tnote{}}
	\label{tab:combined}

	\begin{threeparttable}
	\begin{tabular}{c|c|c|c|c|c|c|c|c|c|c|c|c|c|c|c}
		\toprule
		n & \multicolumn{3}{c|}{256} & \multicolumn{3}{c|}{512} &  \multicolumn{3}{c|}{1024} & \multicolumn{3}{c|}{2048}  & \multicolumn{3}{c}{4096}\\
		\cline{2-16}
		&theo.& exp.& \cite{couteau2018concrete} & theo.& exp. & \cite{couteau2018concrete} & theo.& exp. & \cite{couteau2018concrete} & theo.& exp. & \cite{couteau2018concrete}& theo.& exp. & \cite{couteau2018concrete} \\
		\hline
		$s = 1.45$ &4& 3.6& 4 & 5&4.9 &6 &7&7.0 & 9 &10& 10.1& 14 &15& 14.9& 21 \\
		\hline
		$s = 1.4$ &6&5.7 & 6 &9& 8.4& 11 & 13&12.7 & 17 & 20 &19.3& 27 & 31& 30.0& 44 \\
		\hline
		$s = 1.3$&11 & 10.4 & 20&18 & 16.7 & 23 &30& 27.1& 39 &49 & 44.3 & 65 &80& 72.6& 110 \\
		\bottomrule	 
	\end{tabular}
	\begin{tablenotes}
	\item[]  Columns of theo. and exp. denote the theoretical and experimental results, respectively.  
	\end{tablenotes}
\end{threeparttable}
\end{table*}

We further implement the proposed attack to actually recover the secret as described in Algorithm \ref{alg:guess_and_determine}. The algorithm will not terminate until the guessing values for the guessed variables are correct when enough linear equations are derived. We record the number of guesses for the correct path, the number of nodes we have visited, and the number of times of solving linear systems. Table \ref{tab:key_recovery} shows the results under some $(n, s)$ pairs, each with 40 instances.

\begin{table}[!ht]
	\renewcommand{\arraystretch}{1.2}
	\caption{The Results for Practical Key Recovery \tnote{}}
	\label{tab:key_recovery}
	\centering
	\begin{threeparttable}
	\begin{tabular}{c|c|c|c|c}
		\toprule
		$(n, s)$ & \multicolumn{3}{c|}{Ours} & \cite{couteau2018concrete}\\
		\cline{2-5}
		~~~~~~~~~~&~~~~~\#1~~~~~ &~~~~~\#2~~~~~&~~~~~\#3~~~~~&~~~~~~~~~~   \\
		\hline
		$(256, 1.45)$ & 3.7& 7.4&8.1& 4  \\
		\hline
		$(256, 1.4)$ &6.0& 29.0&31.0 & 6  \\
		\hline
		$(256, 1.3)$ & 11.1& 1038.0&1051.0 & 13\\
		\hline
		$(512, 1.45)$ &5.1 & 18.5&20.0 &6 \\
		\hline
		$(512, 1.4)$ &8.8 & 191.3& 195.0& 11 \\
		\hline
		$(512, 1.3)$ & 17.3& 81562.5&81887.3 & 23 \\
		\hline
		$(1024, 1.45)$ & 6.9& 70.9&73.1& 9 \\
		\hline
		$(1024, 1.4)$ & 13.5& 4366.3&4373.0 & 17 \\
		\hline
		$(2048, 1.45)$ & 10.4& 663.9&668.3 & 14\\
		\hline	
		$(4096, 1.45)$& 15.4& 19357.1&19364.7& 21 \\
		\bottomrule	 	 
	\end{tabular}
	\begin{tablenotes}
	\item[]  The column \#1 and the last column denote the required number of guesses, columns \#2 and \#3 denote the number of times of solving linear systems, and the number of visited nodes, respectively  
\end{tablenotes}
\end{threeparttable}
\end{table}
One can see that the number of visited nodes is slightly larger than the number of times the algorithm is solving linear equation systems, which makes sense since we always immediately trace back whenever a conflict occurs without going into solving a linear system.  The available results match well with the simplified verification results in Table \ref{tab:combined}, from which we could estimate the required numbers of guesses for larger parameters.

The results are better than those in \cite{couteau2018concrete}, particularly when more guesses are required. For example, when $n = 4096, s = 1.3$, 110 guesses are required in \cite{couteau2018concrete}, while our attack only needs around 73. This is because when more guesses are needed, we perform more guesses in Class III and get more ``free'' variables, thus the advantage from exploiting ``free'' variables is more obvious. \\
\\
If a system is expected to achieve a security level of $r$ bits, the complexity for inverting the system should be larger than $2^r$. We use the same estimation of time complexity\footnote{The accurate estimation formula can be found in the proof of concept implementation of~\cite{couteau2018concrete}. See \url{https://github.com/LuMopY/SecurityGoldreichPRG}.} as that in \cite{couteau2018concrete}, which is \(2^{\ell}\cdot n^{2}\), where \(\ell\) and $n$ denote the number of guesses and seed size, respectively. Actually, there should be some constant factor for \(2^{\ell}\cdot n^{2}\) when deriving the actual complexity from the asymptotic complexity $O(2^{\ell}\cdot n^{2})$. However, for a fair comparison, we take the constant factor as 1 as done in \cite{couteau2018concrete}. Thus, to achieve the $r$-bit security,  $\ell$ should satisfy $2^\ell n^2 > 2^r$, i.e., $\ell > r - 2\log_2 n$. Let us set $\ell = r - 2\log_2 n$ and try to derive the limit of stretch $s$ above which the parameters are susceptible to our attack. The unsecure stretch $s$ would satisfy the condition below according to \eqref{eq:number}:
\begin{equation} \label{eq:security_level}
4(r - 2\log_2 n-1)n^{s-1} + c + 2(r - 2\log_2 n) -2  \geq n.
\end{equation}
We derive the theoretical limits for $s$ under 80-bit and 128-bit security levels for different seed sizes $n$ using \eqref{eq:security_level}. We can use the same way to derive the theoretical limits for the guess-and-determine attack in \cite{couteau2018concrete}, which were not given there. It is proved in \cite{couteau2018concrete} that the number of linear equations for guessing one variable is larger than $2\frac{m}{n}$, thus after guessing $\ell$ variables, the number of linear equations is larger than $2\ell\frac{m}{n} + c + \ell$. The unsecure stretch $s$ in \cite{couteau2018concrete} would satisfy the condition below:
\begin{equation} \label{eq:security_level_2018}
	2(r - 2\log_2 n )n^{s-1} + c + r - 2\log_2 n  \geq n.
\end{equation}
Note that the stretch limits in \cite{couteau2018concrete} would be looser than ours, as the theoretical analysis there is based on the worst case, while we consider the average case. 
  
We also get experimental limits by running extensive instances. For a given seed size $n$, we start with a high enough $s$ value and decrease it by a 0.001 interval each time. For each $(n, s)$ pair, we implement 400 independent instances\footnote{We tested a relatively smaller number of instances here, since the cases for $n = 8192$ require much more computation overhead. We found in the experiment that the results averaged on 400 instances and 1000 instances are almost the same, due to the small variances.}, and check whether the complexity is larger than $2^{80}$ or $2^{128}$ under our guess-and-determine attack. We have shown in Table \ref{tab:combined} and Table \ref{tab:key_recovery} that the required numbers of guesses to obtain enough linear equations match well with those to actually recover the secret, and we use the former to represent the latter, since actually solving the linear equation systems consumes large computation overhead. 
\begin{figure}[!ht]
	\centering
	\includegraphics[width=0.5\textwidth]{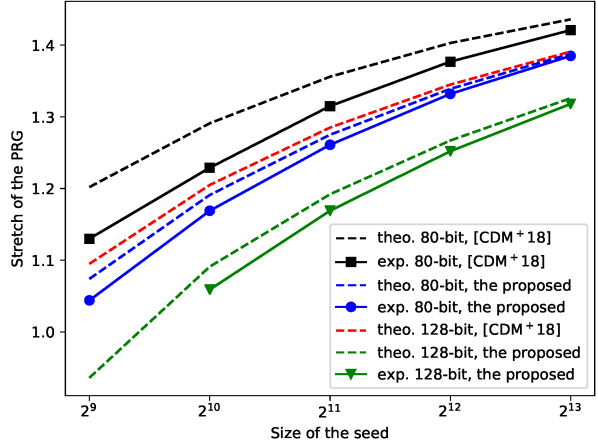}
	\caption{The limits of stretchs for vulnerable instances (theo. and exp. denote the theoretical and experimental results, respectively). The zones above the lines denote the insecure choices of parameters}
	\label{fig:security_level}
\end{figure}
\begin{table}[!t]
	\renewcommand{\arraystretch}{1.2}
	\centering
	\caption{Complexity for Solving the Challenge Parameters Proposed in~\cite{couteau2018concrete}}
	\label{tab:comparison}
	\begin{tabular}{c|c|r|r}
		\toprule
		Security Level & (\(n, s\)) & New &  \cite{couteau2018concrete} \\ 
		\hline
		\addlinespace[1ex]
		80 bits & (\(512, 1.120\)) & \(2^{61}\) & \(2^{91}\) \\
		& (\(1024, 1.215\)) & \(2^{66}\) & \(2^{90}\) \\
		& (\(2048, 1.296\)) & \(2^{68}\) & \(2^{91}\) \\
		& (\(4096, 1.361\)) & \(2^{68}\) & \(2^{91}\) \\
		\hline
		\addlinespace[1ex]
		128 bits & (\(512, 1.048\)) & \(2^{79}\) & \(2^{140}\) \\
		& (\(1024, 1.135\)) & \(2^{93}\) & \(2^{140}\) \\
		& (\(2048, 1.222\)) & \(2^{98}\) & \(2^{139}\) \\
		& (\(4096, 1.295\)) & \(2^{100}\) & \(2^{140}\) \\
		\bottomrule		 
	\end{tabular}
\end{table}

Fig. \ref{fig:security_level} shows the stretch limits for different seed sizes under the 80-bit and 128-bit security levels. The dashed lines denote the theoretical limits for our attack and the attack in \cite{couteau2018concrete}, which are computed according to  \eqref{eq:security_level} and \eqref{eq:security_level_2018}, respectively. The solid lines denote the experimental results (the experimental stretch limit for 128-bit security is not given in \cite{couteau2018concrete}).  The zone above the lines represent the insecure choices of $(n,s)$ parameters. From the results, one can easily see that the stretch limits under our attack are stricter than the ones in \cite{couteau2018concrete}. Thus, some systems which are secure under the attack in \cite{couteau2018concrete} cannot resist against our attack given a security level. Particularly, our stretch limits for 80-bit security are even stricter than the theoretical limit for 128-bit in \cite{couteau2018concrete}, thought the comparison to the experimental limit is unclear. We will later show that we can attack some parameter set suggested for 128-bit security in \cite{couteau2018concrete} with complexity lower than $2^{80}$. Besides, our theoretical and experimental limits have smaller gaps than the ones in \cite{couteau2018concrete}, especially when $n$ goes larger, which indicate that with our theoretical analysis, one can have better predication of the security for a system when the parameters get larger. \\
\\
{\bf Breaking Challenge Parameters.} The authors in~\cite{couteau2018concrete} suggested some challenge parameters for achieving 80 bits and 128 bits of security.  Table~\ref{tab:comparison} compares the time complexities for attaching these parameters. For a fair comparison, we again use the same estimation as that in \cite{couteau2018concrete} where the time complexity is estimated as \(2^{\ell}\cdot n^{2}\), with \(\ell\) being the required number of guesses.
Also, note that the authors of~\cite{couteau2018concrete} took a margin of \(10\%\) when selecting these security parameters. For each $(n, s)$ parameter, we implemented 1040 instances.

We see that all the proposed challenge parameters fail to achieve the claimed security levels and the improvement factor becomes larger when \(n\) is smaller. For \((n,s) = (512, 1.048)\), the time complexity of the improved algorithm is smaller by a factor of about \(2^{61}\); this parameter set is even insufficient for providing \(80\) bits of security, though it is originally suggested for \(128\) bits.

 
\section{Guess-and-Decode: A New Iterative Decoding Approach for Cryptanalysis on Goldreich's PRGs}
\label{sec:guess_decode}

In this section, we present a new attack on Goldreich's PRGs with even lower complexities. We combine the guessing strategies described in Section \ref{sec:method} and iterative decoding to invert a quadratic system and recover the secret, which we call guess-and-decode attack. Specifically, we first guess some variables using similar strategies described in our guess-and-determine attack; then instead of inverting a linear system, we apply probabilistic iterative decoding on the derived quadratic system to recover the secret, for which the complexity is smaller. Besides this gain in reducing the complexity, we also expect that fewer guesses are required before a system can be correctly inverted. 


We first give a recap on the classical iterative decoding showing how it works on linear checks; then describe how we modify it to invert a quadratic system. We experimentally verify the attack, showing that with soft decoding we can further reduce the complexities for attacking the challenge parameters given in \cite{couteau2018concrete}, and finally suggest some new challenge parameters which appear to resist against our attacks for further investigation.
\subsection{Recap on Iterative Decoding} \label{subsec:iter_recap}
Iterative decoding has been commonly used in information theory, e.g., most notably used for decoding LDPC (low-density parity-check) codes. The basic idea of iterative decoding is to break up the decoding problem into a sequence of iterations of information exchange, and after some iterations the system is expected to converge and give a result (or the process is halted).  It can provide sub-optimal performance with a much reduced complexity compared to a maximum likelihood decoder. Below we give a short introduction to LDPC codes and describe how iterative decoding works. For more details about iterative decoding, we refer to \cite{ryan2009channel,hagenauer1996iterative}. 

An $(n, k)$ LDPC code, where $k$ and $n$ respectively denote the lengths of the information block and the codeword, is a linear block code which is usually defined as the null space of a {\it parity-check matrix} $\mathbf{H}$ of size $(n-k) \times n$ whose entries are either 1 or 0. For every valid codeword $\mathbf{v}$, $\mathbf{vH}^T = 0$ is always satisfied. Each row of $\mathbf{H}$ denotes one check and the number of rows indicates the number of checks every codeword should satisfy, while each column denotes one code symbol. If the code symbol $j$ is involved in a check $i$, the entry $(i,j)$ of $\mathbf{H}$ would be 1; otherwise 0. Usually, the density of 1's in $\mathbf{H}$ should be sufficiently low to allow for iterative decoding, thus getting the name LDPC.

A {\it Tanner graph} is usually used to represent a code and help to describe iterative decoding. It is a bipartite graph with one group of nodes being the variable nodes (VNs), i.e., the code symbols, and the other group being the check nodes (CNs). If the variable $j$ is involved in the check $i$, an edge between CN $i$ and VN $j$ is established. The VNs and CNs exchange information along the edges, and the process is referred to as {\it message passing}. Each node (either a variable node or a check node) acts as a local computing processor, having access only to the messages over the edges connected to it. Based on the incoming messages over the connected edges, a node would compute new messages and send them out. In some algorithms, e.g., {\it bit-flipping} decoding, the exchanged messages are binary (hard) values while in others, such as {\it belief-propagation} or {\it sum-product} decoding, the messages are soft probabilities, which represent a level of belief about the values of the code symbols. The probabilities can be implemented in the logarithm domain, i.e., log-likelihood ratios (LLRs), which makes the iterative decoding more stable. We use the sum-product decoding based on LLRs in our attack. 

The sum-product algorithm accepts {\it a-priori} probabilities for the information symbols, which are usually received from the channel and known in advance before the decoding, and outputs {\it a posteriori} values for all symbols after some iterations of message passing. Below we give the general steps of the sum-product algorithm.\\
\\
\noindent{\bf Step 1 Initialization.} For every variable $v$, initialize its LLR value according to its a-priori LLR value $L_a(v)$, i.e., $L_v^{(0)} = L_a(v)$. For every edge connected to $v$, initialize the conveyed LLR value as $L_v^{(0)}$.

After the initialization, the algorithm starts iterations for belief propagation, in which {\bf Step 2 - Step 4} below are iteratively performed until the recovered secret is correct or reaching the maximum allowed number of iterations. In each iteration, every node (either variable node or check node) updates the outgoing LLR value over every edge based on the incoming LLR values over all the {\it other} edges, which are called  ``extrinsic''  information; while ignoring the incoming LLR value over this specific edge, which is called ``intrinsic'' information. It is always the ``extrinsic'' information that should be only used to update the LLR values.\\
\\
\noindent{\bf Step 2 Check node update.} In each iteration $i$, every check node $c$ computes an outgoing LLR value over each of its edges $e_k$, based on the incoming LLR values updated in the $(i-1)$-th iteration from every {\it other} edge $e_k'$ connected to $c$, as below:
\begin{align} \label{eq:check_update}
	L_c^{(i)}(e_k) &= \underset{k' \neq k}{\boxplus} L_v^{(i-1)}(e_{k'})\nonumber \\ &=2\tanh^{-1}\left(\prod_{k' \neq k} \tanh \left(1/2 L_v^{(i-1)}(e_k')\right)\right),
\end{align}
where the "box-plus" operator $\boxplus$ is used to denote the computation of the LLR value of the xor-sum of two variables, i.e., for $a = a_1 \oplus a_2$, $L(a) = L_{a_1} \boxplus L_{a_1} = \log (\frac{1+e^{L_{a_1} + L_{a_2}}}{e^{L_{a_1}} + e^{L_{a_2}}})$.\\
\\
\noindent{\bf Step 3 Variable node update.} In each iteration $i$, every variable node $v$ computes an outgoing LLR value over each of its edges $e_j$, based on the a-priori information, and LLR values updated in the $i$-th iteration from every {\it other} edge $e_j'$ connected to $v$, as below:
\begin{align}\label{eq:variable_update}
	L_v^{(i)}(e_j) =  L_{a}(v)  + \underset{j' \neq j}{\sum} L_c^{(i)}(e_{j'}).
\end{align}

\noindent{\bf Step 4 Distribution update.} After each iteration, update the LLR value of every variable $v$ using \eqref{eq:variable_update}, but with every edge included, i.e.,
\begin{align}\label{eq:distribution_update}
	L_v^{(i)} =  L_{a}(v)  + \underset{j \in N(v)}{\sum} L_c^{(i)}(e_{j}),
\end{align}  
where $N(v)$ is the set of edges connected to $v$. Set
\begin{align} \label{eq:verify}
	\hat{v} = 
	\begin{cases}
		1 & \text{if}~ L_v^{(i)} < 0,\\
		0& \text{otherwise},	
	\end{cases}
\end{align}
to recover an intermediate value for every variable $v$. The algorithm immediately stops whenever  $\mathbf{\hat{v}H^T = 0}$ is satisfied or the number of iterations reaches the maximum limit; otherwise, it continues with a new iteration.

\subsection{Algorithm for the Guess-and-Decode Attack}
In this subsection, we show how we combine guessing and iterative decoding to invert a Goldreich's PRG and recover the secret. We modify the classical iterative decoding to accommodate our use case, and the differences are listed below:
\begin{enumerate}
	\item The most important difference lies in that the system in classical iterative decoding is linear, while quadratic in our application. We have designed special belief propagation techniques for quadratic equations.
	\item In classical iterative decoding, a-priori information, either from a non-uniform source or from the channel output, is required and plays an important role for the convergence of the belief propagation.  However, in our case, there is no available a-priori information and all the variables are assumed to be uniformly random distributed. The system is expected to achieve self-convergence.	
	\item In classical iterative decoding, the check values are always zero; while in our case, a check value could be one, and special belief propagation techniques are designed for it. 
\end{enumerate}
Algorithm \ref{alg:iterative_decoding} shows the general process of the guess-and-decode attack. It basically consists of two phases: guessing phase, during which  guessing strategies similar to what have been described in Section \ref{sec:method} are applied to guess and derive ``free'' variables; and decoding phase, during which the modified iterative decoding is applied on the resulting quadratic system to recover the remaining secret bits. 

\begin{algorithm}[!ht] 
	\caption{The guess-and-decode attack} 
	\textbf{Input} 
	A Goldreich's PRG system instantiated on $P_5$, the maximum allowed number of iterations $iterMax$
	\\	
	\textbf{Output} 
	A recovered secret $\hat{x}$
	\begin{algorithmic}[1] \label{alg:iterative_decoding}
		\STATE Guess some variables using strategies similar to what have been described in Section \ref{sec:method}
		\IF{ the secret can already be correctly recovered}
		\STATE return the recovered secret
		\ELSE	
		\STATE Build an iterative decoding model for the resulting system of equations, initialization, set $it = 0$	
		\WHILE{$it < iterMax$}
		\STATE $it = it + 1$
		\STATE Update check nodes
		\STATE Update variable nodes
		\STATE Update the distributions of variables and get an intermediate  recovered secret $\hat{x}$, return $\hat{x}$ if it is correct.
		\ENDWHILE
		\ENDIF
	\end{algorithmic}
\end{algorithm} 

\subsubsection{Guessing Phase}
The guessing process generally follows the strategies in the guess-and-determine attack described in Section \ref{sec:method}, but with some modifications. Specifically, in Class III, we only guess variables from the linear equations, and if there are no linear equations, we guess a variable in equations in Class II (or further Class I if Class II is empty), instead of guessing a variable in a quadratic equation in Class III. This is because we want to keep the quadratic equations and get biased information for some  involved variables. For example, for an equation $x_{\sigma_1} + x_{\sigma_2}x_{\sigma_3} = y$, after the first iteration of belief propagation, we would get a biased distribution of $x_{\sigma_1}$, i.e., $P(x_{\sigma_1} = 1) = 0.25, P(x_{\sigma_1} = 0) = 0.75$. Such biased information could propagate through the graph and help to make the system converge. This is how the system achieves self-convergence without any a-priori information.  

It should be mentioned here that after guessing a relatively large number of variables, it could happen that all the remaining variables can be directly recovered without going into iterative decoding. This happens in two cases:
\begin{itemize}
	\item[(1)] after guessing enough variables, all the other variables could be determined for free, which happens with a non-negligible probability;
	\item[(2)] besides the guessed and freely determined variables, every remaining variable is involved as the linear term in an quadratic equation in Class III and can be recovered correctly with high probability. For example, for an equation $x_{\sigma_1} + x_{\sigma_2}x_{\sigma_3} = y$, we can recover $x_{\sigma_1}$ as $x_{\sigma_1} = y$. This happens more often when more variables are guessed. The experimental results would verify this later. 
\end{itemize} 

Thus, after the guessing phase, we always check if all the variables can already be correctly recovered; and if so, the iterative decoding can be omitted.

\subsubsection{Decoding Phase}
When the guessing phase is done, we build an iterative decoding model for the derived system of equations and start the belief propagation. \\
\\
{\bf Iterative decoding model.} The remaining unknown secret variables (neither guessed nor freely determined) are modeled as the variable nodes, while all the remaining valid equations, either quadratic or linear, serve as the checks. Recall that we could have six different forms of equations in the derived system which we have categorized into three classes in Section \ref{subsec:equation_classes}, along with one more form of those which only have a quadratic term, e.g., $x_{\sigma_1^i}x_{\sigma_2^i} = 0$. 
A check node and the variable nodes which are involved in this given check would be connected through edges in the Tanner graph. We define two different types of edges: {\it type 1} edge, which connects a check node and a variable node that is involved as a linear term of this check; and {\it type 2} edge, which connects a check and a variable that is involved in the quadratic term of this check. For these two different types of edges, different belief propagation techniques are applied, which we would show below. Besides, the check value of a check could be either 1 or 0, which does not happen in the classical iterative decoding, and we would describe how we deal with it in a moment. 

\begin{figure*}[!ht]
	\centering
	\includegraphics[width=0.80\textwidth]{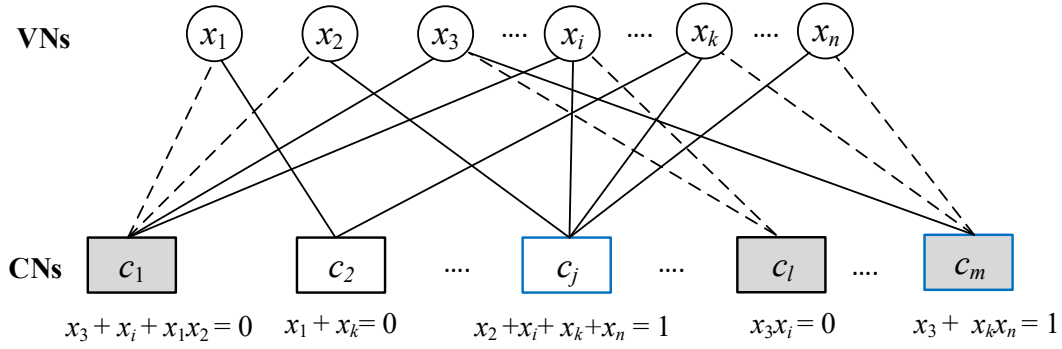}
	\caption{Illustration of an example of the iterative decoding model}
	\label{fig:model}
\end{figure*}
Fig. \ref{fig:model} shows a simple illustration of our iterative decoding model, where circles denote the variable nodes while rectangles denote the check nodes. The gray-filled checks have quadratic terms where the solid lines denote {\it type 1} edges while the dashed lines denote {\it type 2} edges. The blue-bordered checks are those with check values being one.
The decoding phase generally follows the routine of classical iterative decoding described in Section \ref{subsec:iter_recap}, but we have some novel modifications for, e.g., dealing with checks with quadratic terms or with check values being one. We next describe these details in each step.\\
\\
\noindent{\bf Step 1 Initialization.} Since there is no a-priori information for the variables, the LLR values of all the variables are initialized to be zero. All the outgoing messages, i.e., LLR values, over the edges of variables or checks are set to be zero as well. 

After the initialization, the algorithm starts iterations for belief propagation. The update of variable nodes follows the classical iterative decoding, while for updating the check nodes, we have special belief propagation techniques, which are described in details below.\\
\\
\noindent{\bf Step 2 Check node update.} The update of a linear check just follows the classical way, with a small modification when the check value is one. We below give a simple example to show how it affects the belief propagation.\\
\\
\noindent {\it Example.} Given a weight-2 check $x_{\sigma_1} + x_{\sigma_2}  = y$, assume the incoming LLR value in a certain iteration for $x_{\sigma_2}$ is $L_{x_{\sigma_2}}$.
\begin{itemize}
	\item[--]If $y = 0$, the combinations of ($x_{\sigma_1}, x_{\sigma_2}$) to validate the check is $(0, 0), (1, 1)$. Thus, 
	\begin{align*}
		&P(x_{\sigma_1} = 0) \\
		&= P(x_{\sigma_1} = 0, x_{\sigma_2} = 0) + P(x_{\sigma_1} = 0, x_{\sigma_2} = 1) \\
		&=  P(x_{\sigma_1} = 0| x_{\sigma_2} = 0) P(x_{\sigma_2} = 0) + 0\\
		& = P(x_{\sigma_2} = 0).
	\end{align*}
	Similarly, we can get $P(x_{\sigma_1} = 1) = P(x_{\sigma_2} = 1)$. Thus, the outgoing LLR value of $x_{\sigma_1}$ value is computed as $L_{x_{\sigma_1}} = \log \frac{P(x_{\sigma_2} = 0)}{P(x_{\sigma_2} = 1)} = L_{x_{\sigma_2}}$.
	\item[--]While if $y = 1$, the combinations of ($x_{\sigma_1}, x_{\sigma_2}$) to validate the check is $(0, 1), (1, 0)$. Similarly we would get $P(x_{\sigma_1} = 0) =  P(x_{\sigma_2} = 1)$, $P(x_{\sigma_1} = 1) =  P(x_{\sigma_2} = 0)$. The outgoing LLR value of $x_{\sigma_1}$ is then computed as $L_{x_{\sigma_1}} = \log \frac{P(x_{\sigma_2} = 1)}{P(x_{\sigma_2} = 0)} = -L_{x_{\sigma_2}}$. 
\end{itemize}

Thus, for a linear check, when the check value is zero, the outgoing LLR values are computed with the standard way in \eqref{eq:check_update}; while when it is one, the negative versions of values computed using \eqref{eq:check_update} are sent. \\
\\
When updating the quadratic check nodes, different belief propagation techniques are applied for {\it type 1} and {\it type 2} edges. The update for {\it type 1} edges follows the classical way using \eqref{eq:check_update}, while including the equivalent incoming LLR value from the quadratic term as well. For {\it type 2} edges, we should deal with them carefully by distinguishing ``intrinsic'' and ``extrinsic'' information. Below we show how to update a quadratic check node.

For a quadratic check, e.g., $x_{\sigma_1} + x_{\sigma_2} + x_{\sigma_3} + x_{\sigma_4}x_{\sigma_5}  = y$, assume the incoming LLR values over the edges are $L_{x_{\sigma_1}}, L_{x_{\sigma_2}}, L_{x_{\sigma_3}}, L_{x_{\sigma_4}}, L_{x_{\sigma_5}}$, respectively, and we want to compute the outgoing LLR value over each edge. We denote the linear part and quadratic part as $x_l$ and $x_q$, respectively, i.e., $x_l = x_{\sigma_1} + x_{\sigma_2} + x_{\sigma_3}$, $x_q = x_{\sigma_4}x_{\sigma_5}$. We could compute the equivalent LLR values for $x_l$ and $x_q$, denoted as $L_{x_l}, L_{x_q}$, respectively.  For $x_l$, we could easily get
 $L_{x_l} = L_{x_{\sigma_1}} \boxplus L_{x_{\sigma_2}} \boxplus L_{x_{\sigma_3}}$. While for $x_q$, we fist get
\begin{align}
	p_{x_q}^1& = p_{x_{\sigma_4}}^1 p_{x_{\sigma_5}}^1 = \frac{1}{(e^{L_{x_{\sigma_4}}} + 1)(e^{L_{x_{\sigma_5}}} + 1)},\nonumber \\
	p_{x_q}^0 &=  1 - p_{x_{\sigma_4}}^1 p_{x_{\sigma_5}}^1 = \frac{(e^{L_{x_{\sigma_4}}} + 1)(e^{L_{x_{\sigma_5}}} + 1) - 1}{(e^{L_{x_{\sigma_4}}} + 1)(e^{L_{x_{\sigma_5}}} + 1)}.
\end{align}
Thus the equivalent incoming LLR value for $x_q$ can be computed as $L_{x_q} = \log \frac{p_{x_q}^0}{p_{x_q}^1} = \log ((e^{L_{x_{\sigma_4}}} + 1)( e^{L_{x_{\sigma_5}}} + 1) - 1)$. If $y = 0$, the outgoing LLR value of a linear variable, say $x_{\sigma_1}$, can be computed as:
\begin{align}
	L_{x_{\sigma_1}} = L_{x_{\sigma_2}} \boxplus L_{x_{\sigma_3}} \boxplus L_{x_h}. 
\end{align} 
While if $y=1$, the negative version, i.e., $-L_{x_{\sigma_1}} $ should be sent, just like the update of a linear check when its check value is one. The outgoing LLR values for $x_{\sigma_2}, x_{\sigma_3}$ can be computed in the same way. 

Next we show how to compute the outgoing LLR values for the variables in the quadratic term, e.g., $ x_{\sigma_5}$.  For combinations of $(x_l, x_{\sigma_4}, x_{\sigma_5})$, when $y = 0$, the possible values to validate the check are $(0,0,0)$, $(0,0,1)$, $(0,1,0)$, $(1,1,1)$. Thus the updated LLR value for $x_{\sigma_5}$ is 
\begin{align}
	L_{x_{\sigma_5}} &= \log \frac{p_{x_l}^0 p_{x_{\sigma_4}}^0 p_{x_{\sigma_5}}^0 + p_{x_l}^0  p_{x_{\sigma_4}}^1 p_{x_{\sigma_5}}^0 }{p_{x_l}^0 p_{x_{\sigma_4}}^0 p_{x_{\sigma_5}}^1 + p_{x_l}^1 p_{x_{\sigma_4}}^1 p_{x_{\sigma_5}}^1}\\
	&=\log \frac{p_{x_{\sigma_5}}^0}{p_{x_{\sigma_5}}^1} + \log \frac{p_{x_l}^0 }{p_{x_l}^0p_{x_{\sigma_4}}^0 + p_{x_l}^1p_{x_{\sigma_4}}^1}. \label{eq:extrinsic}
\end{align}
In \eqref{eq:extrinsic}, the first term is regarded as the ``intrinsic'' information while the second term is the ``extrinsic'' information which should be propagated. Thus the outgoing LLR value for $x_{\sigma_5}$ is 
\begin{align} \label{eq:quad_out}
	L_{x_{\sigma_5}} &= \log \frac{p_{x_l}^0 }{p_{x_l}^0p_{x_{\sigma_4}}^0 + p_{x_l}^1p_{x_{\sigma_4}}^1}.
\end{align}
While if $y = 1$,  the possible values of $(x_l, x_{\sigma_4}, x_{\sigma_5})$ to validate the check are $(1,0,0)$, $(1,0,1)$, $(1,1,0)$, $(0,1,1)$. $L_{x_{\sigma_5}}$ can still be computed using \eqref{eq:quad_out}, but with the values of $p_{x_l}^0$ and $p_{x_l}^1$ being exchanged. We could understand it as that $L_{x_l}$ now has the negative version. The outgoing LLR value for $x_{\sigma_4}$ can be derived in the same way. 

The update techniques apply to all types of checks: a linear check and a check with only a quadratic term can be regarded as special cases without $x_q$ and $x_l$, respectively. For the latter, e.g., $x_{\sigma_1} x_{\sigma_2} = 0$, the possible combinations of $(x_{\sigma_1}, x_{\sigma_2})$ to validate the check is $(0,0),(0,1),(1,0)$, thus the LLR value of $x_{\sigma_1}$ can be updated as $L_{x_{\sigma_1}}  = \log (1/{p_{x_{\sigma_2}}^0}) = \log (1 + e^{-L_{x_{\sigma_2}}}).$ The LLR value of $x_{\sigma_2}$ can be updated in the same way.\\ 
\\
\noindent {\bf Important Notes.} We mentioned before that no a-priori information is required for the convergence in our iterative decoding model. Instead, after the first iteration of updating the check nodes, some biased information for some variables would be obtained. The biased information mainly comes from quadratic checks in Class III and quadratic checks without any linear terms. As we mentioned, for a quadratic check in  Class III $x_{\sigma_1} + x_{\sigma_2}x_{\sigma_3} = 0$,  $x_{\sigma_1}$ can become highly biased soon: i.e., $P(x_{\sigma_1} = 0) = 0.75, P(x_{\sigma_1} = 1) = 0.25$ and the outgoing LLR value over the edge would become $\log 3$ instead of zero. Similarly, for a check $x_{\sigma_1}x_{\sigma_2} = 0$, the outgoing LLR values for $x_{\sigma_1}, x_{\sigma_2}$ would become $\log 2$. Such biased information can then be propagated and spread to other nodes during the iterations, which plays an important role for the convergence and correctness of the iterative decoding. Obviously, the more accurate biased information is obtained, the more likely the system would be to correctly converge, which explains why we try to avoid guessing a variable in a quadratic equation in Class III during the guessing phase.

\medskip

\noindent{\bf Steps 3, 4} just follow the details described in {\bf Steps 3, 4} in Section \ref{subsec:iter_recap}.
\subsection{Theoretical Analysis}
A deeper theoretical investigation into the complexity is hard, as we have much more complicated model compared to classical iterative decoding: we have irregular checks which can be linear or quadratic, and have different localities; besides, there is no available a-priori information. Instead, we derive a very rough estimation for the asymptotic complexity. Under the same amount of guesses as that in the guess-and-determine attack, the guess-and-decode attack would succeed with a large probability. When one guessing path results in a failure, several other paths can be tried and it is highly likely that one would succeed. We allow for a constant number of iterations, and in each iteration, we need to go through $O(n^s)$ nodes. Thus the rough asymptotic complexity can be expressed as $O(n^s2^{\frac{n^{2-s}}{4}})$.
\subsection{Experimental Verification}
We have done extensive experiments to verify the attack. We first test the success probabilities of correctly recovering a secret under different numbers of guesses $\ell$ for the challenge parameters suggested in \cite{couteau2018concrete}, averaged over 5000 independent instances for each $(n,s,\ell)$ parameter set. 

\begin{figure*}[!ht]
	\centering
	\includegraphics[width=0.9\textwidth]{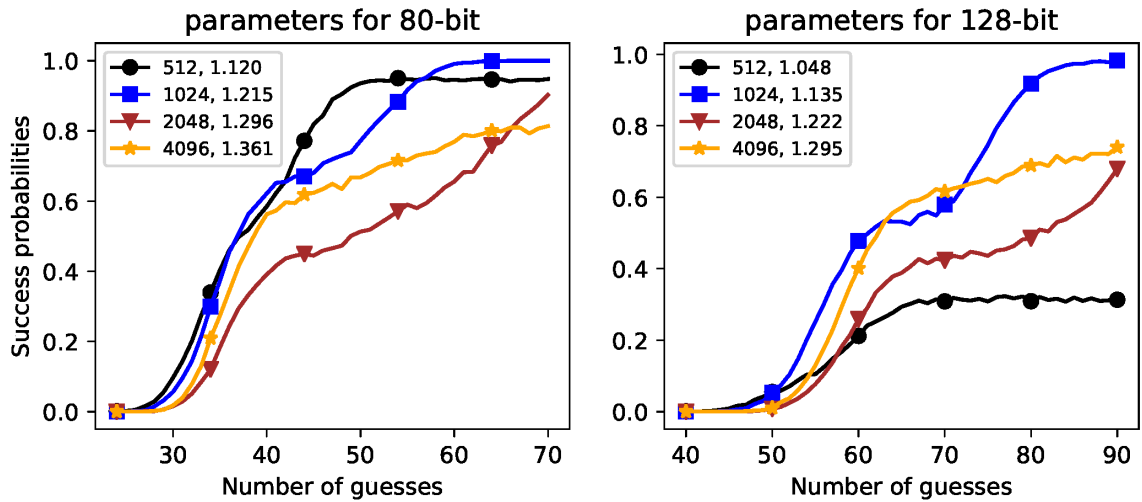}
	\caption{Success probabilities under different number of guesses}
	\label{fig:probability}
\end{figure*}

Fig. \ref{fig:probability} illustrates the results of success probabilities when maximally 100 iterations are allowed, where the left and right sub-figures are for the parameters suggested for 80-bit and 128-bit security levels in \cite{couteau2018concrete}, respectively. 
One can see that the success probabilities increase with the number of guesses, though sometimes with very small fluctuations. We found in the results that under different numbers of guesses, the required numbers of iterations for the convergence of a system vary: basically, when more variables are guessed, fewer iterations are required. Particularly, when guessing more than some certain number of variables, a secret could be recovered without going into iterations of belief propagation in some instances, which could happen with two cases as we mentioned before. 

\begin{figure*}[!ht]
	\centering
	\includegraphics[width=0.85\textwidth]{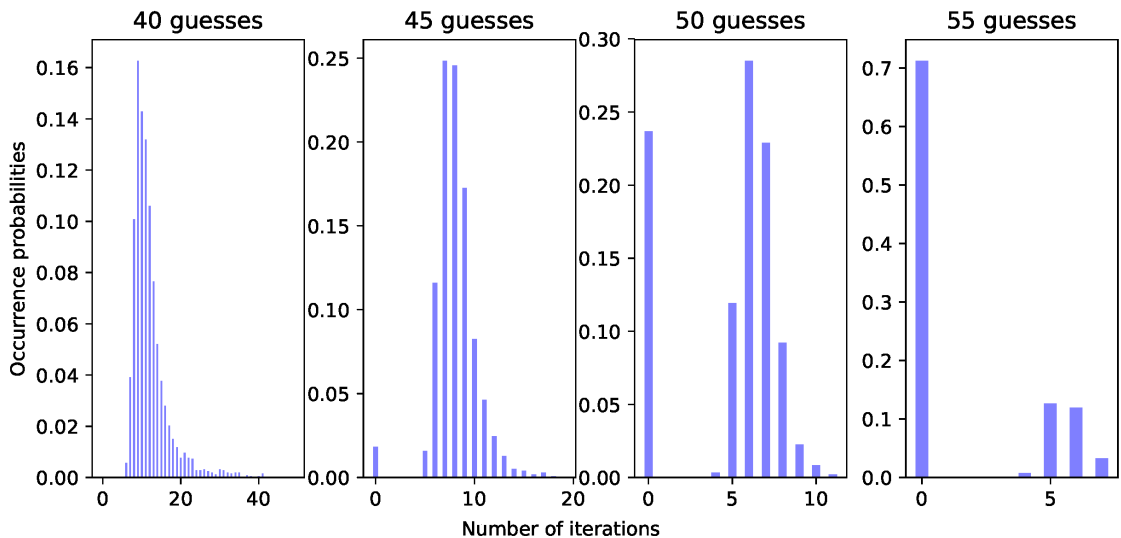}
	\caption{Distributions of number of iterations for successful instances}
	\label{fig:different_distri}
\end{figure*}

Fig. \ref{fig:different_distri} shows the  distributions of the required numbers of iterations for convergence when guessing 40, 45, 50 and 55 variables, respectively, when $n = 1024, s = 1.215$. The instances with zero iterations denote those which can be solved without iterative decoding. We can see that with the increase in the number of guesses, the proportion of these instances is increasing: when guessing 40 variables, all the successful recoveries are solved through iterative decoding; while when guessing 55 variables, 71.0\% of the successful instances can be solved without going into iterative decoding.   

We also see from Fig. \ref{fig:different_distri} that the required numbers of iterations for the vast majority of instances are below 100, and mostly below 40, and the average required number of iterations decreases when the number of guesses increases. This applies to all instances under different parameter sets. In Fig. \ref{fig:distri_iterations} we give an example to further illustrate this observation. The left sub-figure shows the distribution of the required number of iterations for 97259 independent successfully inverted instances when  guessing 26 to 51 variables for $n = 1024, s = 1.215$ and 44 to 69 variables for $n = 1024, s = 1.215$\footnote{We did not consider the required numbers of iterations when guessing more variables since most would be zero and here we only want to show  that the required numbers of iterations are mostly below 100.}.
The right sub-figure shows the average required number of iterations under different numbers of guesses. Basically, the required number of iterations decreases when the number of guesses increases (with some exceptions in the beginning due to the small number of samples), which is because more equations with lower localities could be obtained when more variables are guessed, enabling faster convergence. Specifically, when guessing more than 53 and 75 variables when $s = 1.215$ and $s = 1.135$, respectively, the required numbers of iterations for most instances are zero, meaning that we can recover the secret directly without iterative decoding.


\begin{figure*}[!ht]
	\centering
	\includegraphics[width=0.9\textwidth]{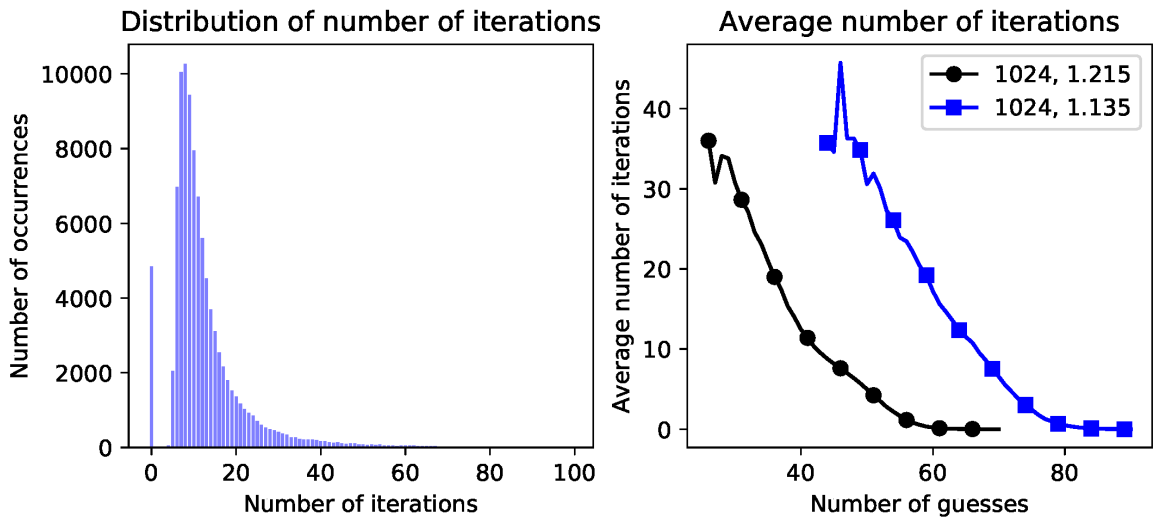}
	\caption{More results about the required number of iterations}
	\label{fig:distri_iterations}
\end{figure*}

We also tested the distribution of number of guesses required to correctly recover a secret. We ran 5000 independent instances for parameter sets $n = 1024,s = 1.215$, and $n = 1024,s = 1.135$, which respectively aim for 80-bit and 128-bit security levels in \cite{couteau2018concrete}. We start from guessing a small number of variables, and continue with guessing one more variable if the recovery fails, until the secret is correctly recovered. The results are shown in Fig. \ref{fig:num_guess_distri}. We can obviously see from the figure that instances aiming for 80-bit security level require fewer guesses than those for 128-bit security level, for which the peak values are respectively achieved around 36 and 59 guesses. 
\begin{figure}[!ht]
	\centering
	\includegraphics[width=0.5\textwidth]{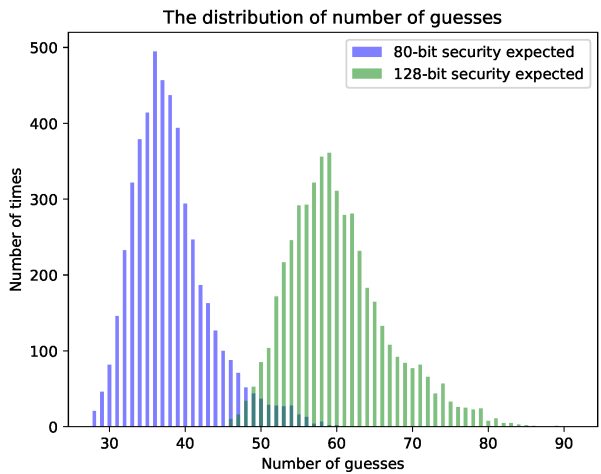}
	\caption{Distribution of number of guesses for $n = 1024$}
	\label{fig:num_guess_distri}
\end{figure}
\subsection{Complexity} \label{subsec:simulated_complexity}
In our guessing strategies, we always prefer to guess the variable which appears most often in a local class or in the global system, which is a greedy method. Actually, if we choose to guess the sub-optimal ones, the success probabilities do not have big differences, as the main gain comes from exploiting ``free'' variables. This enables us to explore other guessing paths if the current one fails for iterative decoding.

If the success probability of the guess-and-decode attack when guessing $\ell$ variables for a given $(n, s)$ parameter set is $p$, the required number of times to perform iterative decoding could be derived by multiplying a factor of $1/p$. This can be linked to two practical attacking scenarios: 1) $1/p$ many independent systems generated by different secrets are collected and there should be one system that can be recovered; 2) or given a specific system, it is possible to recover it if we try  $1/p$ independent guessing paths. We have experimentally verified this. For example, given $n = 1024, s = 1.215$ and maximally allowed 40 iterations, the success probability when guessing 32 variables is 0.1226, which means we need to try 8.2 ($1/0.1226$) independent paths on average for attacking a system. We generated 5000 independent instances of parameter set $n = 1024, s = 1.215$ and for each instance, we continue choosing a guessing path and performing iterative decoding until we finally correctly invert the system. We computed the average number of independent guessing paths we need to try, which is 9.2, matching well with the result computed from success probabilities.  

We have shown in Fig. \ref{fig:distri_iterations} that the vast majority of instances require less than 40 iterations for convergence, thus we set the maximum allowed number of iterations as 40. For each iteration, we need to update each variable node and check node, which requires complexity $O(n^s)$. The update of each variable node is simple, which only costs several additions, as can be seen from \eqref{eq:variable_update}. For a check node, let us consider a worst case of complexity for updating, i.e., considering a check with most edges, e.g., $x_{\sigma_1} + x_{\sigma_2} + x_{\sigma_3} +x_{\sigma_4}x_{\sigma_5} = y$. For updating such a check, it requires 3 exponentiation, 5 division, 4 logarithm, 13 multiplication, 4 tangent and 4 inverse tangent  functions. These functions roughly take 1300 clock cycles if we refer to the instruction manuals for mainstream CPUs. The actual cost would be less, as there are many simpler checks which require fewer clock cycles. For example, the update of an equation check of the form $x_{\sigma_1} + x_{\sigma_2} = y$ is almost for free. We have experimentally verified the cost with our non-optimized code. For example, one iteration for a system of parameters $n = 512, s = 1.120$ consumes 0.46 milliseconds{\footnote{This includes all the clock cycles for one iteration, i.e., updating the check nodes and variable nodes, updating the distributions of variables, recovering an intermediate secret, etc.}}, running on a CPU with maximum supported clock speed 3000 MHz, corresponding to 1275 clock cycles for updating one node. There are many optimization techniques both for the iterative decoding algorithm and for the implementation details (e.g., using parallelization, using look-up tables, etc), which would largely reduce the required number of clock cycles. We believe that the constant factor is not larger than that in $O(n^2)$ for inverting a sparse matrix. 
Thus, the total complexity for inverting an instance of parameter set $(n, s, \ell)$  can be derived as below, up to some constant factor:
\begin{equation} \label{eq:iter_complexity}
	\frac{1}{p} 2^{\ell}\cdot 40 \cdot n^s.
\end{equation}


Having made this clear, we could get the complexities for attacking the systems under the challenge parameters suggested in \cite{couteau2018concrete}. For every $(n, s)$ parameter set, we generate 5000 independent instances{\footnote{We generate 10 times more instances when $n=512$ and 1024, and the success probabilities do not have big differences as those computed from 5000 instances. Thus we conjecture that 5000 instances are large enough for computing a stable success probability.}} and run the attack with maximally 40 iterations allowed when guessing $\ell$ variables for different $\ell$ values. We compute the complexities using \eqref{eq:iter_complexity} according to the experimental success probabilities  for different $\ell$, and choose the smallest one as the complexity for breaking the given $(n, s)$ pair. For $n = 512$, we have some additional results: we see from Fig. \ref{fig:probability} that the success probabilities are consistently increasing and when we guess a large enough number of guesses $\ell'$, most instances can be recovered without iterative decoding. If we denote the success probability of such case as $p'$, the complexity would be $1/p' \cdot 2^{\ell'}\cdot n${\footnote{We also checked the complexities using this way for other parameter sets, but they are all higher than the ones derived through iterative decoding.}} (with $n$ included since we need to go through every secret variable.).

\begin{table*}[!ht]
	\renewcommand{\arraystretch}{1.15}
	\centering
	\caption{The Simulated Complexities for the Challenge Parameters Proposed in \cite{couteau2018concrete}}
	\label{tab:breaking}
	\begin{threeparttable}
	\begin{tabular}{c|c|c|c|c} 
		\toprule
		Security Level & (\(n, s\)) & Number of Guesses & Success (probabilities) &Complexity \\ 
		\hline
		\addlinespace[1ex]
		80 bits & (\(512, 1.120\))$^*$ & $40$ & 718 (0.1446) &$2^{52}$\\ 
		& (\(1024, 1.215\)) & 32 & 613 (0.1226) & $2^{53}$\\
		& (\(2048, 1.296\)) & 32 & 210 (0.042) &$2^{57}$\\
		& (\(4096, 1.361\)) & 32 & 285 (0.057) &$2^{58}$\\
		\hline
		\addlinespace[1ex]
		128 bits & (\(512, 1.048\))$^*$ & 53 & 86 (0.0172) & $2^{68}$ \\
		& (\(1024, 1.135\)) & 50 & 209 (0.0418) &  $2^{72}$\\
		& (\(2048, 1.222\)) & 52 &82 (0.0164) &$2^{77}$\\
		& (\(4096, 1.295\)) & 51 & 85 (0.0166) &$2^{78}$\\
		\bottomrule		 
	\end{tabular}
	\begin{tablenotes}
	\item[] The complexities for parameters with $^*$ are derived based on the instances for which iterative decoding is not needed.  
\end{tablenotes}
\end{threeparttable}
\end{table*}

Table \ref{tab:breaking} shows the complexities along with the corresponding required numbers of guesses and success probabilities. Actually, the complexities vary slow with $\ell$, and there are several $\ell$ values under which the complexities are the same to the best ones shown in the table. We can see the complexities are much lower than the claimed ones in  \cite{couteau2018concrete}. Particularly, the parameter sets suggested for 128 bits security in  \cite{couteau2018concrete} cannot even provide 80 bits security under our attack.  

\subsection{New Challenge Parameters}
\label{subsec:new_parameters}
We also suggest some practical range of parameter sets which appear to be resistant against both the proposed guess-and-determine attack and the guess-and-decode attack for 80 and 128 bits security. As done in \cite{couteau2018concrete}, we take a \(10\%\) margin to select the security parameters. Under a certain seed size $n$, we vary $s$ from a high enough value to lower ones with a relatively larger interval 0.01 (for instances of $n = 512$,  we take a smaller interval 0.001), since we actually ran the iterative decoding process and recovered the secret in our implementation which requires more running time. We choose the maximum $s$ value for which the system is not vulnerable to our attack as the conjectured stretch limit. 

Given each $(n, s)$ pair, we vary the number of guesses $\ell$ from a low enough value to a high enough value and run 5000 instances for each $(n, s, \ell)$ parameter set. We record the success probability for each $\ell$ value and further compute the corresponding complexity using \eqref{eq:iter_complexity}. For $n = 512$, we further compute the complexity for recovering the secret without using iterative decoding. We then chose the minimum one as the complexity for the $(n, s)$ pair and check if it is larger than $2^{80}$ or $2^{128}$. We chose the maximum $s$ for which the condition is satisfied as the conjectured stretch limit for challenge parameters.

Table \ref{tab:challenge_parameter_2} shows the results under different seed sizes. Further study is required to guarantee confidence in the security levels given by these parameters. One can see that the stretch limits are further narrowed down with a large gap from the ones in \cite{couteau2018concrete}. For example, when $n = 1024$, the results in \cite{couteau2018concrete} suggest that the systems with stretches smaller than 1.215 and 1.135 can provide 80 and 128 bits security, respectively, while our results show that the stretches should be smaller than 1.08 and 1.02, respectively. Besides, our results show that systems with seed sizes of 512 are not suitable for constructing local PRGs: they cannot even provide 80 bits security.
\begin{table}[!ht]
	\renewcommand{\arraystretch}{1.2}
	\centering
	\caption{Challenge Parameters for Seed Recovery Attack}
\label{tab:challenge_parameter_2}
	\begin{tabular}{c c c c c}
		\toprule
		\hline
		security level & ~~512~~&~~1024~~&~~2048~~& ~~4096~~\\ 
		\hline
		$80$ &-& 1.08&1.18 & 1.26\\
		\hline
		$128$ & -& 1.02& 1.10 & 1.19 \\
		\bottomrule	 
	\end{tabular}\\
\end{table}  
\section{Extension to Other Predicates} \label{sec:extension}
It is interesting and of importance to know if the proposed attacks apply to other predicates and which ones are susceptible to or resistant against them. In this section, we investigate this question and focus on the two main types of predicates suggested for constructing local PRGs, i.e., XOR-AND and XOR-THR (threshold) predicates, which are defined as below:
\begin{align*}
	&\text{XOR}_k-\text{AND}_{q}:x_1+\cdots+x_k + x_{k+1}x_{k+2}\cdots x_{k+q},\\
	&\text{XOR}_k-\text{THR}_{d,q}: x_1+\cdots+x_k + \text{THR}_{d,q}(x_{k+1},\dots,x_{k+q}),
\end{align*}
where $\text{THR}_{d,q}(x_{k+1},\dots,x_{k+q})$ is a threshold function of which the value would be one only when the number of one's in $(x_{k+1},\dots,x_{k+q})$ is not less than $d$; otherwise, the value is zero. $P_5$ can be regarded as a special case of XOR-AND and XOR-THR predicates. These predicates have been investigated in \cite{couteau2018concrete, applebaum2016cryptographic, applebaum2018algebraic, meaux2019improved}.

\subsection{Extensions to Other XOR-AND Predicates}
In the survey paper \cite{applebaum2016cryptographic}, the authors asked the question ``{\it{Is it possible to efficiently invert the collection $\mathcal{F}_{P, n, m}$ for every predicate $P$ and some  $m = n^{\frac{1}{2}\lfloor {2d/3} \rfloor-\epsilon }$ for some $\epsilon > 0$?}} '' and gave ``a more concrete challenge'': XOR$_k$-{AND}$_q$ predicates with $k = 2q$. We consider these challenging predicates and investigate their resistance against our attacks. 

The guessing phase is similarly to what has been described in the guess-and-determine attack. Below we derive the belief propagation techniques for a general XOR$_k$-{AND}$_q$ predicate. 


For an XOR$_k$-{AND}$_q$ equation $x_1 + x_2 + \cdots + x_k + x_{k+1}x_{k+2} \cdots x_{k+q} = y$, suppose the incoming LLR values for $x_1, x_2,...,x_{k+q}$ are $L_{x_1}, L_{x_2},\dots, L_{x_{k+q}}$, respectively. Let $x_h$ denote the AND term, i.e.,  $x_h = x_{k+1}x_{k+2} \cdots x_{k+q}$, such that we get $P(x_h^1 ) = p_{x_{k+1}}^1 p_{x_{k+2}}^1\cdots p_{x_{k+q}}^1$, $P(x_h^0)= 1-P(x_h^1 )$. The equivalent incoming LLR value of $x_h$,  denoted as $L_{x_h}$, can be computed as: 
\begin{align}
&L_{x_h} = \log \frac{1-p_{x_{k+1}}^1 p_{x_{k+2}}^1\cdots p_{x_{k+q}}^1}{ p_{x_{k+1}}^1 p_{x_{k+2}}^1\cdots p_{x_{k+q}}^1}\nonumber \\
&= \log \left((1 + e^{L_{x_{k+1}}})(1 + e^{L_{x_{k+2}}})\cdots(1 + e^{L_{x_{k+q}}}) - 1\right).
\end{align}
The outgoing LLR value of a linear term, say  $x_1$, can thus be computed as:
\begin{align}
L_{x_1} = L_{x_2} \boxplus L_{x_3} \boxplus \cdots \boxplus L_{x_h}. 
\end{align}
Similarly, if $y=1$, $L_{x_1}$ should be the negative value. The outgoing LLR values for other linear terms can be derived in the same way. We next show how to update the LLR value of a variable involved in the AND term, say $x_{k+1}$ without loss of generality. We denote the LLR value of the linear part $x_l= x_1 + x_2 +\cdots+ x_k$ as $L_{x_l}$, which can be computed as $L_{x_l} =  L_{x_1} \boxplus L_{x_2} \boxplus \cdots \boxplus L_{x_k}$.

Denote the product of other variables in the AND term excluding $x_{k+1}$ as $x_h'$, i.e., $x_h' = x_{k+2}x_{k+3}\cdots x_{k+q}$. For the combinations $(x_l, x_{k+1}, x_h')$, the possible values to validate the check are $(y, 0,0)$, $(y, 0,1)$, $(y, 1,0)$, $(y+1, 1,1)$. The updated LLR value for $x_{k+1}$ can then be computed as: 
\begin{align} \label{eq:update_xor_and}
L_{x_{k+1}} &= \log \frac{p_{x_l}^y p_{x_{k+1}}^0 p_{x_h'}^0 + p_{x_l}^y  p_{x_{k+1}}^0 p_{x_h'}^1 }{p_{x_l}^y p_{x_{k+1}}^1 p_{x_h'}^0 + p_{x_l}^{y+1} p_{x_{k+1}}^1 p_{x_h'}^1} \nonumber \\
&=\log \frac{p_{x_{k+1}}^0}{p_{x_{k+1}}^1} + \log \frac{p_{x_l}^y }{p_{x_l}^y p_{x_h'}^0 + p_{x_l}^{y+1} p_{x_h'}^1}.
\end{align}
With the ``intrinsic'' part excluded, i.e., the first term in the second line of \eqref{eq:update_xor_and}, the outgoing LLR value of $x_{k+1}$ can be computed as
\begin{align}\label{eq:out_xor_and}
L_{x_{k+1}} &= \log \frac{p_{x_l}^y }{p_{x_l}^y p_{x_h'}^0 + p_{x_l}^{y+1} p_{x_h'}^1},
\end{align}
where $p_{x_{h'}}^1 = p_{x_{k+2}}^1p_{x_{k+3}}^1...p_{x_{k+q}}^1$ and $p_{x_{h'}}^0 =  1-p_{x_{h'}}^1$. For an equation having only the higher-degree term, e.g., $ x_{k+1}x_{k+2} \ldots x_{k+q} = 0$, we get $L_{x_{k+1}} = \log \frac{1}{p_{x_h'}^0}$. 

\noindent{\bf Experimental results} We investigate several challenging XOR$_k$-AND$_q$ predicates with $k = 2q$, but for relatively small localities, since iterative decoding typically applies to sparse systems. Specifically, we investigate the following three concrete predicates:
\begin{align}
&\text{XOR}_4-\text{AND}_2: x_1 +x_2+x_3 + x_4 + x_5x_6,\\
&\text{XOR}_6-\text{AND}_3: x_1 + x_2+ \cdots + x_6 + x_7x_8x_9,\\
&\text{XOR}_8-\text{AND}_4: x_1 + x_2 + \cdots+ x_8 + x_9x_{10}x_{11}x_{12}.
\end{align}
Fig. \ref{fig:extensions} shows the success probabilities of the guess-and-decode attack applied on these three predicates under certain stretches. Basically, the attack works better when the locality is lower. For example, the attack applies to XOR$_4$-AND$_2$ well, while not so good to XOR$_6$-AND$_3$ and XOR$_8$-AND$_4$ predicates. We could get the attacking complexities for these parameters using \eqref{eq:iter_complexity}, and results show that XOR$_4$-AND$_2$ under parameter sets $(512, 1.3)$ and $(1024, 1.3)$ cannot provide security levels of 80 bits and 128 bits, respectively; XOR$_6$-AND$_3$ under $(512, 1.4)$ cannot achieve 128 bits security; while the attacking complexity for XOR$_8$-AND$_4$ under $(512, 1.45)$ is $2^{130}$.
\begin{figure}[!ht]
	\centering
	\includegraphics[width=0.46\textwidth]{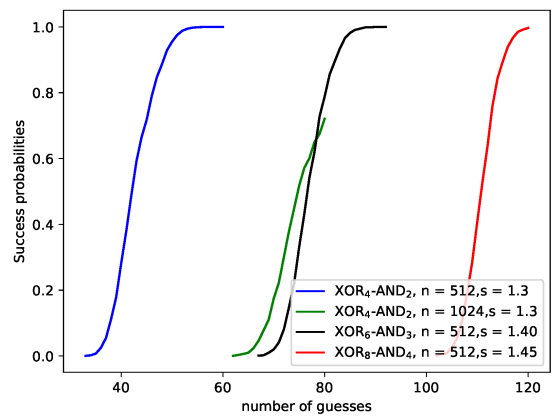}
	\caption{Success probabilities for XOR-AND predicates under different parameters}
	\label{fig:extensions}
\end{figure}

We further investigate the stretch regimes above which the predicates cannot achieve 128 bits security following the way in Section \ref{subsec:new_parameters}. For each $(n, s, \ell)$ set, we run 2000 instances. Table \ref{tab:weak_xor_and} shows the results of stretch limits above which the corresponding predicates under the given seed sizes cannot provide 128 bits security. We failed to get the results when $n \geq 2048$ (and $n = 1024$ for XOR$_8$-AND$_4$ predicate) due to the large computation overhead, as the $s$ limits would be much larger. One can see from the results that these predicates have better resistance than $P_5$ against our attack.

\begin{table}[!t]
	\renewcommand{\arraystretch}{1.2}
	\centering
	\caption{The Stretch Limits for Different XOR-AND Predicates}
\label{tab:weak_xor_and}
		\begin{threeparttable} 
	\begin{tabular}{c c c c} 
		\toprule
		\hline
		$n$ & XOR$_4$-AND$_2$ & XOR$_6$-AND$_3$ & XOR$_8$-AND$_4$ \\
		\hline
		512& < 1 &  1.17 & 1.46\\ 
		\hline
		1024& 1.16 &  1.51& -\\ 
		\bottomrule		 
	\end{tabular}
	\begin{tablenotes}
	\item[] Stretches above these values are vulnerable.  
\end{tablenotes}
\end{threeparttable}
\end{table}

\subsection{Extensions to XOR-THR Predicates}
In \cite{applebaum2018algebraic}, the authors show that predicates with high resiliency and high degrees are not sufficient for constructing local PRGs. They proposed a new criterion called rational degree against algebraic attacks and suggested XOR-MAJ (majority) predicates, which can be generalized to the XOR-THR predicates.
In \cite{meaux2019improved}, the properties of XOR-THR predicates are investigated and the special case XOR-MAJ predicates are used to build the new version of FLIP, FiLIP, which is a construction designed for homomorphic encryption. It is mentioned in the paper that ``no attack is known relatively to the functions XOR$_k$-THR$_{d,2d}$ or XOR$_k$-THR$_{d,2d-1}$ since $k \geq 2s$ and $d \geq s$''.  These predicates are actually XOR-MAJ predicates.

We try to apply the guess-and-decode attack to these predicates. For the guessing phase, the cases how free variables are obtained differs a bit compared to the XOR-AND predicates. Specifically, there are two cases a free variable could be obtained:
\begin{itemize}
	\item[(1)]For an equation with one linear term and a THR term, if the value of the THR term is known, for example, there are already not less than $d$ one's or more than $q - d$ zero's  in the THR term, the linear term can be freely derived;
	\item[(2)]For an equation which only has a THR term, if there already exist $d-1$ one's in the THR term while the value for the equation is zero, every other variable could be derived as zero for free; on the other hand, if there are $q - d$ zero's while the value of the equation is one, all the other variables could be derived as one for free.
\end{itemize}

After the guessing phase, the iterative decoding is performed. 
For an equation $ x_1 +\cdots+x_k + \text{THR}(x_{k+1}, \ldots, x_{k+q})=y$, suppose the incoming LLR values of the variables are $L_{x_1},L_{x_2},\ldots, L_{x_{k+q}}$, respectively. When computing the outgoing LLR values for the linear terms, say $x_1$ without loss of generality, we need to first compute the equivalent incoming LLR value of the THR term. We denote the THR term as $x_t$ and its LLR value as $L_{x_t}$. The number of the combinations of  $(x_{k+1}, \ldots, x_{x+q})$ that make $\text{THR}(x_{k+1}, \ldots, x_{x+q})$ one is $\sum_{w = d}^{q} \binom{q}{d}$. If we denote the set of these combinations as $\mathcal{S}$, we can get
\begin{align} \label{eq:subset_probability}
P(x_t^1) = \sum_{(c_{k+1}, \ldots c_{x+q}) \in \mathcal{S}} \prod_{i = 1}^{q} p_{x_{k+i}}^{c_{k+i}},
\end{align}
where $(c_{k+1}, \ldots, c_{x+q})$ are the possible values in $\mathcal{S}$ and every value is either one or zero. $P(x_t^0) =  1- P(x_t^1)$ and thus $L_{x_t} = \log \frac{P(x_t^0)}{P(x_t^1)}$. The LLR value of $x_1$ can be computed as $ L_{x_1} = L_{x_2} \boxplus \ldots \boxplus L_{x_k} \boxplus L_{x_t}$.

If there are some variables in the $\text{THR}(x_{k+1}, \ldots x_{x+q})$ term being fixed, e.g., guessed or determined during the guessing phase, we can just fix its probability of being the fixed value as one, while zero for the complement value.  

We next show how to compute the outgoing LLR value of a variable in the $\text{THR}(x_{k+1}, \ldots x_{x+q})$ term, say $x_{k+1}$ without loss of generality. Let $x_l$ denote the linear part, i.e., $x_l = x_1 + ...+x_k$, then the LLR value of it, denoted $L_{x_l}$, is derived as $L_{x_l} = L_{x_1} \boxplus \ldots \boxplus L_{x_k}$. Denote the set of combinations of $(x_{k+2}, x_{k+3},\ldots, x_{x+q})$ that have no less than $d-1$ one's as $\mathcal{W}$, and $\bar{\mathcal{W}}$ for the complement set. $P(W), P(\bar{W})$ can be computed using the same way as in \eqref{eq:subset_probability}. Then we would get:
\begin{align}
	&P(x_{k+1} = 1) \nonumber \\
	&= p(x_{k+1} = 1, x_l =  y+1, \mathcal{W} ) + p(x_{k+1} = 1, x_l = y, \bar{\mathcal{W}}) \nonumber \\
	&= p_{x_{k+1}}^1\cdot  p_{x_{l}}^{y+1}\cdot P(\mathcal{W}) + p_{x_{k+1}}^1 \cdot p_{x_{l}}^y\cdot P(\bar{\mathcal{W}}) 
\end{align}

Denote the set of combinations of $(x_{k+2}, x_{k+3},\ldots, x_{x+q})$ that have no less than $d$ one's as $\mathcal{V}$, and $\bar{\mathcal{V}}$  for the complement set. We can get the probability of $p(x_{k+1} = 0)$ using the same way as below:
\begin{align}
	&P(x_{k+1} = 0) \nonumber \\
	&= p(x_{k+1} = 0, x_l = y+1, \mathcal{V} ) + p(x_{k+1} = 0, x_l = y, \bar{\mathcal{V}})\nonumber \\
	&=p_{x_{k+1}}^0 \cdot p_{x_{l}}^{y+1} \cdot P(\mathcal{V}) + p_{x_{k+1}}^0 \cdot p_{x_{l}}^y \cdot P(\mathcal{\bar{V}})
\end{align}

With the ``intrinsic'' part being excluded, the outgoing LLR value of $x_{k+1}$ can be computed as 
\begin{align}
	L_{x_{k+1}} = \frac{ p_{x_{l}}^{y+1} \cdot P(\mathcal{V}) +  p_{x_{l}}^y \cdot P(\mathcal{\bar{V}})}{p_{x_{l}}^{y+1}\cdot P(\mathcal{S}) + p_{x_{l}}^y\cdot P(\bar{\mathcal{S}}) }.
\end{align}
The LLR values of other variables in the THR term can be derived using the same way. One can see that when $q$ is large, computing the combinations would require large overhead, introducing better resistance against our attack. Thus we only consider four concrete XOR-THR predicates, which are actually XOR-MAJ predicates: XOR$_3$-MAJ$_3$, XOR$_3$-MAJ$_4$, XOR$_4$-MAJ$_3$, XOR$_4$-MAJ$_4$.

\begin{figure}[!ht]
	\centering
	\includegraphics[width=0.48\textwidth]{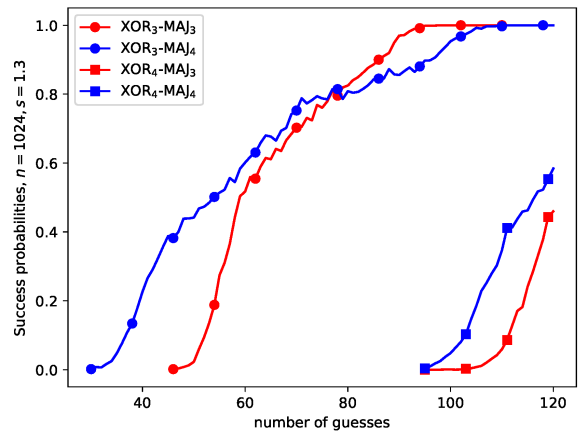}
	\caption{Success probabilities for XOR-MAJ predicates for $n = 1024, s = 1.3$}
	\label{fig:success_xor_maj}
\end{figure}

Fig. \ref{fig:success_xor_maj} shows the success probabilities under different numbers of guesses when $n = 1024, s = 1.3$. We can see that our guess-and-decode attack applies to these predicates well. Using \eqref{eq:iter_complexity}, we can get that when $n = 1024, s = 1.3$, local PRGs instantiated on XOR$_4$-MAJ$_3$, XOR$_4$-MAJ$_4$ predicates  cannot achieve 128 bits security, while XOR$_3$-MAJ$_3$, XOR$_3$-MAJ$_4$ cannot even achieve 80 bits security.\\
\\
{\bf Discussion.} One important observation from Fig. \ref{fig:success_xor_maj} is that the number of linear terms matters more than the number of variables in the MAJ term. For example, the differences of success probabilities between predicates with a same number of linear terms, e.g., XOR$_4$-MAJ$_3$ and XOR$_4$-MAJ$_4$, are much smaller than those between predicates with a same number of variables in the MAJ term, e.g., XOR$_3$-MAJ$_3$ and XOR$_4$-MAJ$_3$. This makes sense since iterative decoding applies well to sparse systems, and it is more sensitive to the number of terms in a check: the MAJ term can be regarded as one special term. Another interesting observation is that XOR$_i$-MAJ$_3$, $i \in [3,4]$ are more resistant against our attack compared to XOR$_i$-MAJ$_4$, $i \in [3,4]$. This could be because that the predicates are balanced when $q = 2d -1$,  while not so when $q = 2d$. It is pointed out in \cite{meaux2019improved} that the {\it resiliency} of a XOR$_k$-THR$_{d,q}$ is $k$ if $q = 2d -1$, while $k-1$, otherwise. The predicates used in the suggested FiLIP instances in \cite{meaux2019improved} are all of the type XOR$_k$-THR$_{d,2d-1}$.  Our results match well with the analysis in \cite{meaux2019improved} and serve as a direct illustration of how the {\it resiliency} of a predicate would affect its security. 

We further ran extensive experiments to get the limit stretches above which the local PRGs instantiated on different predicates are susceptible to our attack for 128 bits security. Table \ref{tab:weak_xor_maj} shows the results. One can see that size 512 is not suitable for constructing efficient local PRGs instantiated on these given predicates.

\begin{table}[!ht]
	\renewcommand{\arraystretch}{1.2}
	\centering
	\caption{The Stretch Limits for Different XOR-THR Predicates}
\label{tab:weak_xor_maj}
	\begin{threeparttable}
	\begin{tabular}{c c c c c} 
		\toprule
		\hline
		$n$ & XOR$_3$-MAJ$_3$ & XOR$_3$-MAJ$_4$ & XOR$_4$-MAJ$_3$& XOR$_4$-MAJ$_4$\\
		\hline
		512& <1 & < 1  & 1.03 & 1.02\\ 
		\hline
		1024& 1.10  & 1.06 & 1.30& 1.26\\ 
		\hline
		2048 & 1.26 & 1.18 & 1.48 & 1.46 \\
		\bottomrule		 
	\end{tabular}
	\begin{tablenotes}
	\item[] Stretches above these values are vulnerable. 
\end{tablenotes}
\end{threeparttable}
\end{table}

\section{Concluding Remarks} \label{sec:conclusion}
We have presented a novel guess-and-determine attack and a guess-and-decode attack on Goldreich's pseudorandom generators instantiated on the \(P_{5}\) predicate, greatly improving the attack proposed in \cite{couteau2018concrete}. Both attacks work based on similar guessing strategies: we try to explore as many variables which can be determined for free as possible. In the guess-and-decode attack, we use a modified iterative decoding method to solve the resulting system after guessing a certain number of variables, which provides a new idea to solve a sparse quadratic system. We broke the candidate non-vulnerable parameters given in \cite{couteau2018concrete} with a large gap and suggested some new challenge parameters which could be targets for future investigation. 

The attacks further narrow the concrete stretch regime of  Goldreich's pseudorandom generators instantiated on the \(P_{5}\) predicate and largely shake the confidence in their efficiency when the seed sizes are small. 

We further extend the attacks to investigate some other predicates of the XOR-AND and XOR-MAJ type, which are suggested as research target for constructing local PRGs. Generally, local PRGs instantiated over predicates with low localities show susceptibility to our attacks. Especially, the number of linear terms plays an important role in resisting against our attacks. It is good to have more than six linear terms if large stretches are desired. For the non-linear part of a predicate, it is better to be balanced than being unbalanced to resist against our attacks. If in some extreme cases local PRGs instantiated on predicates with low localities are required, our attacks could be helpful for choosing a safe stretch. 

The attack might apply to other predicates with similar structures as long as they have low localities.


\bibliographystyle{IEEEtran}
\bibliography{p5}
%



\ifCLASSOPTIONcaptionsoff
  \newpage
\fi

\end{document}